%% file: stab.tex
\documentclass[final]{elsarticle}

\usepackage[utf8]{inputenc} 
\usepackage[english]{babel} 
\usepackage{fullpage}

\usepackage[normalem]{ulem}

\usepackage{amsfonts}
\usepackage{amsmath}
\usepackage{amssymb}
\usepackage{amsthm}
\usepackage{algorithm}
\usepackage[noend]{algpseudocode}
\usepackage{subfigure,epsfig,color}
\usepackage{colortbl}
\usepackage{tikz}
\usetikzlibrary{mindmap}

\usepackage{xifthen}
\usepackage{geometry}
\usepackage{pdflscape}
\usepackage{enumitem}

\newcommand{\ed}[1]{\textcolor{black}{#1}}
\newcommand{\ee}[1]{\textcolor{black}{#1}}
\newcommand{\ef}[1]{\textcolor{black}{#1}}

\newtheorem{lemma}{Lemma}
\newtheorem{theorem}{Theorem}
\newtheorem{definition}{Definition}
\newtheorem{corollary}{Corollary}

\tikzstyle{vertex}=[circle,minimum size=7.50pt]

\begin{document}

\title{General Cut-Generating Procedures for the Stable Set
Polytope\footnote{This work has been partially supported by the Stic/AmSud joint
program by CAPES (Brazil), CNRS and MAE (France), CONICYT (Chile) and MINCYT (Argentina) --project 13STIC-05-- and the Pronem program by FUNCAP/CNPq (Brazil) --project ParGO.}}

\author[ufrrjdcc]{Ricardo C. Corr\^{e}a\corref{ufc}}
\ead{correa@ufrrj.br}
\author[sarmiento1]{Diego Delle Donne}
\ead{ddelledo@ungs.edu.ar}
\author[sarmiento2]{Ivo Koch}
\ead{ikoch@ungs.edu.ar}
\author[sarmiento1]{Javier Marenco}
\ead{jmarenco@ungs.edu.ar}

\address[ufrrjdcc]{Universidade Federal Rural do Rio de Janeiro,
Departamento de Ciência da Computa\c c\~ao,
Av. Governador Roberto Silveira S/N,
26020-740 Nova Iguaçu - RJ, Brazil}
\address[sarmiento1]{Universidad Nacional de General Sarmiento,
Instituto de Ciencias,
J. M. Guti\'errez 1150, Malvinas Argentinas, (1613) Buenos Aires,
Argentina}
\address[sarmiento2]{Universidad Nacional de General Sarmiento,
Instituto de Industria,
J. M. Guti\'errez 1150, Malvinas Argentinas, (1613) Buenos Aires,
Argentina}

\cortext[ufc]{This author was with the ParGO research team, Universidade
Federal do Cear\'a, Brazil, when most part of this work has been done.}

\begin{abstract}
We propose general \ee{separation} procedures for generating cuts for the
stable set polytope, inspired by a procedure by Rossi and Smriglio and applying
a lifting method by Xavier and Camp\^{e}lo. In contrast to existing
cut-generating procedures, ours generate both rank and non-rank valid
inequalities, hence they are of a more general nature than existing methods.
This is accomplished by iteratively solving a lifting problem, which consists of
a maximum weighted stable set problem on a smaller graph. Computational
experience on DIMACS benchmark instances shows that the proposed approach may be
a useful tool for generating cuts for the stable set polytope.
\end{abstract}

\maketitle

\section{Introduction}

Let $G = (V, E)$ be an undirected graph with node set $V$ and edge set $E$. A
\emph{stable set} in $G$ is a subset of pairwise non-adjacent vertices of $G$.
Given a graph $G$, the \ee{\emph{maximum stable set (MSS)}} problem asks
for a stable set $S$ in $G$ of maximum cardinality. The {\em stability number of
$G$} is \ee{this maximum cardinality} and is denoted by $\alpha(G)$. The
MSS problem is \ef{computationally hard to solve in practice, being in NP-Hard}
\ee{unless the graph $G$ has some special structure. For an arbitrary input graph
$G$, a number of exact methods have been developed to solve it through several
combinatorial or mathematical programming-based techniques. For a survey
of these theoretical and practical aspects of the MSS problem,
see~\cite{Bomze99themaximum} and references therein}.

\ef{Enumerative combinatorial algorithms have shown to be efficient to solve the
MSS problem exactly for moderately sized graphs (for an overview,
see~\cite{WuHao.15}). Typically, such algorithms perform a search in a tree with
the employment of simple and fast, but still effective, bounding procedures for pruning
purposes. In this vein, the most successful approach involves the use of
approximate colorings of selected subgraphs of $\bar G$ (the complement of $G$).
This bound is based on the following remark: if $\bar G$ admits an
$\ell$-coloring, then $\alpha(G) \leq \ell$. This is a relatively
weak bound and, consequently, the procedure to compute it is generally applied
at numerous nodes of the search tree. However, it can be computed quickly by
means of a greedy coloring heuristic implemented with bit parallelism operations~\cite{CorreaMCMD14,Segundo.Losada.Jimenez.11,
Tomita.Kameda.07}.}

\ef{An alternative to combinatorial algorithms is the use of sophisticated
mathematical programming techniques to handle the combinatorial properties of
the polytope associated with the formulation $\alpha(G) = \max \{ \sum_{v
\in V} x_v \mid x_u + x_v \leq 1, uv \in E, x_v \in \{ 0, 1\}, v \in V\}$.}
Although combinatorial methods for the MSS problem \ee{from the literature} outperform mathematical programming-based algorithms \ee{devised so far}, it is of great interest to continue the search for
efficient polyhedral methods for this problem. \ee{Despite its natural
theoretical relevance, there are other motivations of algorithmic nature,
namely: (a) the algorithms can be easily extended to the weighted version of
the MSS problem}, (b) MSS constraints frequently appear as a sub-structure in
many combinatorial optimization problems, \ee{(c) in many situations, probing
strategies gives MSS valid inequalities on conflicting variables for general
mixed integer programs (see, e.g., \cite{Atamturk200040,BritoSantosPoggi15})},
and (d) real applications may need specific versions of the MSS problem with
additional constraints. \ee{In this context, procedures for valid inequalities
generation often turns out to be effective}.

\ef{There are two main directions of research when polyhedral techniques, in
particular procedures for valid inequalities generation, are concerned. The
first direction is usually referred to as the {\em lift-and-project}
method~\cite{Cornuejols.08}, which consists in three steps: first, a lifting
operator is applied to the initial formulation to obtain a lifted formulation in
a higher dimensional space; second, the lifted formulation is strengthened by
means of additional valid inequalities; and third, a strengthened relaxation
of the initial formulation is finally obtained as a result of an appropriate
projection of the strengthened lifted formulation onto the original space.
Several upper bounds for $\alpha(G)$ have been stated in connection of this
method, such as the ones based on semidefinite programming (SDP, for short)
relaxations described in~\cite{Dukanovic.Rendl.07,Lovasz.Schrijver.91}, which can be rather
time-consuming to compute in practice. More recently, a new relaxation was
introduced in~\cite{Giandomenico.Rossi.Smriglio.13} which preserves some
theoretical properties of SDP relaxations in generating effective cuts but is
computationally more tractable for a range of synthetic instances.}

\ed{The second direction of research is integer programming, which in turn have
followed two main approaches. The first approach consists in developing strong cuts coming from facet-inducing inequalities associated with special structures in the input graph (\ee{for instance,} cliques, odd holes, webs, among others) and searching for specialized separation techniques for these families of
inequalities \ee{(an up to date list of references for this approach can be
found in~\cite{Pardalos})}.
The second approach relies on general cut-generating procedures which, starting from a fractional solution, search for a violated
inequality with no prespecified structure. Such procedures were either shown to
be effective in practice \cite{Pardalos, Rossi.Smriglio.01} or to generate provably
strong inequalities \cite{XavierCampelo11}. The main contribution of this work are
general procedures that are both effective and generate inequalities
that can be proved to be facet-inducing under quite general conditions.}

\ed{We now discuss existing works following the second approach, i.e., procedures
generating cuts with no prespecified structure.}
Mannino and Sassano~\cite{ManninoSassano96} introduced in 1996 the idea of edge
projections as a specialization of Lov\'asz and Plummer's clique projection
operation~\cite{LovaszPlummer86}. Many properties of edge projections are
discussed in~\cite{ManninoSassano96} and, based on these properties, a procedure
computing an upper bound for the MSS problem is developed. This bound is then
incorporated in a branch and bound scheme. Rossi and Smriglio take these ideas
into an integer programming environment in~\cite{Rossi.Smriglio.01}, where a
separation procedure based on edge projection is proposed. Finally, Pardalos et al.~\cite{Pardalos} extend
the theory of edge projection by explaining the facetness properties of the
inequalities obtained by this procedure. The authors give a branch and cut
algorithm that uses edge projections as a separation tool, as well as several
specific families of valid inequalities such as the odd hole inequalities (with
a polynomial-time exact separation algorithm), the clique inequalities (with \ed{heuristic separation procedures}), and mod-$\{2, 3, 5, 7\}$ cuts.

Rossi and Smriglio propose in \cite{Rossi.Smriglio.01} to employ a sequence of
edge projection operations to reduce the original graph $G$ and make it denser
at the same time, allowing for a faster identification of clique inequalities on
the reduced graph $G'$. \ed{This procedure iteratively removes and projects
edges with certain properties, and heuristically finds violated rank
inequalities (i.e., inequalities of the form $\sum_{v\in A} x_v \leq
\alpha(G[A])$, where $A\subseteq V$ and $G[A]$ is the subgraph of $G$ induced by
$A$).} A key step for achieving this is to be able to establish how $\alpha(G)$
is affected by these edge projections, or, in other words, how exactly
$\alpha(G)$ relates to $\alpha(G')$. We aim at generalizing Rossi and Smriglio's
procedure by projecting cliques instead of edges, so we also need to show how
$\alpha(G)$ changes as a result of this operation. Our method allows thus to
establish a more general relation between $G$ and \ee{the graph resulting from
the clique projection}.

In this article we propose the use of clique projections as a general method for
cutting plane generation for the MSS, along with new clique lifting
\ee{operations that lead} to stronger inequalities than those obtained with the
edge projection method. The proposed method is able to generate both rank and
\ee{weighted} rank valid inequalities (to be defined below), by resorting to the
general lifting \ee{operation} introduced in~\cite{XavierCampelo11}. This
approach allows to produce cuts of a quite general nature, including cuts from
the known families of valid inequalities for the MSS polytope. \ee{Based
upon the projection and lifting operations, we give a separation procedure that}
departs from the usual template-based paradigm for generating cuts, and seeks to
unify and generalize the separation procedures for the known cuts. In this
sense, our main goal is to provide a more complete understanding of the maximum
stable set polytope, which may help also in the solution of other combinatorial
optimization problems. \ee{Experimental results are provided to validate the
general procedure we propose.}

This work is organized as follows. In Section~\ref{sec:polytope} we define the
MSS polytope $STAB(G)$, \ed{we define the operation of clique projection and we
explore some basic facts on this operation. Section~\ref{sec:lifting} introduces the
crucial concept of \emph{clique lifting}, based on the results in
\cite{XavierCampelo11}.} In Sections~\ref{sec:procedures} and~\ref{sec:separation} we
introduce our cut-generating method, by applying the lifting method presented
in~\cite{XavierCampelo11}. Finally, in Section~\ref{sec:experiments} we present some
computational experience on the DIMACS and randomly generated instances, which
show that the method is competitive. \ee{The paper is closed with some
concluding remarks in Section~\ref{sec:conc}.}

\input{procedures}
\input{impl-exp}

\section{Conclusions}
\label{sec:conc}

In this work we have presented general cut-generating procedures for the
standard formulation of the maximum stable set polytope, which are able to
generate both violated rank and generalized rank inequalities. The main
objective of these procedures is to generalize existing ones based on edge
projection, and employ a lifting procedure in order to construct general valid
inequalities from an initial clique inequality by undoing the operation of
clique projection in the original graph. The computational experiments show that
the proposed procedures are effective at generating general cuts, and may be
competitive in a general setting.

\bibliographystyle{plain}
\bibliography{./stab}

\newpage
\appendix

\input{sufficient}

\end{document}

%% file: procedures.tex
\section{The Stable Set Polytope and the Clique Projection Operation}
\label{sec:polytope}

Let $n := |V|$, \ee{$N_G(v)$ be the neighborhood of $v$ in graph $G$,} and
$\mathcal S(G)\subseteq\{0,1\}^n$ be the set of all characteristic vectors of
stable sets of $G$. We write simply \ee{$N(v)$ and} $\mathcal S$,
\ee{respectively}, when $G$ is clear from context. For $W \subseteq V$,
$\mathcal S(G[W])$ stands for the characteristic vectors of stable sets of $G$ involving vertices in $W$ only. The polytope of stable sets of $G$ is denoted by
\[
STAB(G) = \mbox{conv} \{ x \mid x \in \mathcal S(G) \}.
\]
Note that the stability number of $G$ is $\alpha(G) = \max \{ \sum_{v \in V} x_v \mid x \in STAB(G) \}$. If $c \in \mathbb{R}^n$, then the {\em weighted \ed{stability} number of $G$, according to $c$} is $\alpha(G, c) = \max \{ c^\top x \mid x \in STAB(G) \}$. The general form of a facet-inducing inequality of $STAB(G)$ is
\begin{equation}
c^\top x \leq \alpha(G[H], c),
\label{eq:ineq}
\end{equation}
where $c \in \mathbb{R}^n$, $c \geq \mathbf{0}$, $H = \{ v \in V \mid c_v > 0
\}$, and $(G[H], c)$ is a \ed{so-called} {\em facet-subgraph} of $G$~\cite{LiptakLovasz01}.
Note that if $c\in\{0,1\}^n$ then we have the rank inequality mentioned in the Introduction.

Our interest is to build inequalities of type~\eqref{eq:ineq} by means of the
following operation.

\begin{definition}[Clique
Projection~\cite{LovaszPlummer86}] Let $W \subseteq V$, $|W| \geq 2$, be a
clique in $G$. The {\em clique projection} of $W$ gives the graph $G \mid W = (V, E \mid W)$ in which $E\mid W
= E \cup \{ uv \not\in E \mid W \subseteq N(u) \cup N(v) \}$.
\end{definition}

\ed{Figure~\ref{fig:projection} shows an example of this operation.}
The edges in $(E \mid W) \backslash E$ (i.e., the added edges after the
projection) are called \emph{false edges}. These are the edges simulated
by $W$ in the sense stated in the lemma below. \ed{For $W\subseteq V$, we
define $x_W = \sum_{v\in W} x_v$.}

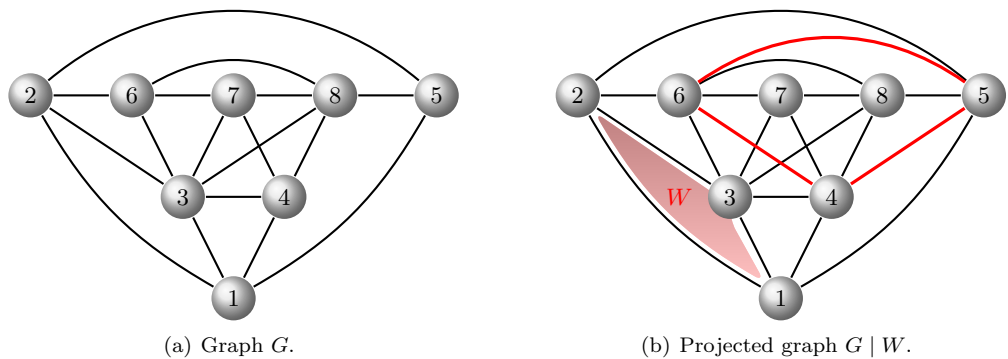
\begin{figure}[htb]
\centering
        \begin{subfigure}[Graph $G$.]{
			\quad\begin{tikzpicture}[thick,fill opacity=0.5,scale=0.45,font=\small]
\foreach \name/\xpos in {2/-6.0,6/-3.0,7/0,8/3.0,5/6.0}
	\node[vertex,ball color=gray!20,opaque] (\name) at (\xpos,3.0) {$\name$};
\foreach \name/\xpos in {3/-1.5,4/1.5}
	\node[vertex,ball color=gray!20,opaque] (\name) at (\xpos,0) {$\name$};
\node[vertex,ball color=gray!20,opaque] (1) at (0,-3.0) {1};

\path[-] 	(8) 	edge (5)
					edge (7)
			(2) 	edge (6)
			(6) 	edge (7)
			(1) 	edge (3)
			(1) 	edge (4)
			(2) 	edge (3)
			(3) 	edge (6)
			(3) 	edge (4)
			(3) 	edge (8)
			(4) 	edge (8)
			(7) 	edge (3)
			(4) 	edge (7);

\draw		(1) 	to[out=150,in=-60] (2);
\draw		(1) 	to[out=30,in=-120] (5);
\draw		(6)		to[out=30,in=150] (8);
\draw		(2) 	to[out=40,in=140] (5);
			\end{tikzpicture}\quad}
        \end{subfigure}\quad\quad
        \begin{subfigure}[Projected graph $G \mid W$.]{
\begin{tikzpicture}[thick,fill opacity=0.5,scale=0.45,font=\small]
\shade[rounded corners=1ex,fill=red,above] (1.north west) to[out=150,in=-60]
(2.south east) -- (3.base) -- (3.south) -- cycle;
\node[text=red,text opacity=1] at +(-3,0) {$W$};

\foreach \name/\xpos in {2/-6.0,6/-3.0,7/0,8/3.0,5/6.0}
	\node[vertex,ball color=gray!20,opaque] (\name) at (\xpos,3.0) {$\name$};
\foreach \name/\xpos in {3/-1.5,4/1.5}
	\node[vertex,ball color=gray!20,opaque] (\name) at (\xpos,0) {$\name$};
\node[vertex,ball color=gray!20,opaque] (1) at (0,-3.0) {1};

\path[-] 	(8) 	edge (5)
					edge (7)
			(2) 	edge (6)
			(6) 	edge (7)
			(1) 	edge (3)
			(1) 	edge (4)
			(2) 	edge (3)
			(3) 	edge (6)
			(3) 	edge (4)
			(3) 	edge (8)
			(4) 	edge (8)
			(7) 	edge (3)
			(4) 	edge (7);

\draw		(1) 	to[out=150,in=-60] (2);
\draw		(1) 	to[out=30,in=-120] (5);
\draw		(6)		to[out=30,in=150] (8);
\draw		(2) 	to[out=40,in=140] (5);

\draw[very thick,color=red]		(5) 	to[out=145,in=35] (6);
\draw[very thick,color=red]		(4) 	-- (5);
\draw[very thick,color=red]		(4) 	-- (6);
\end{tikzpicture}}
        \end{subfigure}
\caption{Projected graph $G \mid W$ (with 3 false
edges) is obtained after projecting $W = \{ 1, 2, 3 \}$.}
\label{fig:projection}
\end{figure}

\begin{lemma}
$F_W = \{ x \in STAB(G) \mid x_W = 1 \} \subseteq STAB(G \mid W)
\subseteq STAB(G)$.
\label{lem:FSTABG}
\end{lemma}

\begin{proof}
For the first \ed{inclusion}, we show that $x \in (STAB(G) \setminus STAB(G \mid
W)) \cap \mathcal S(G)$ implies $x \not\in F_W$. For such an $x$, there is a
false edge $wz \in (E \mid W) \setminus E$ such that $x_{\{ w \}} = x_{\{ z \}}
= 1$. By definition of clique projection, $W \subseteq N(w) \cup N(z)$. Hence,
for every $v \in W$, either $v \in N(w)$ or $v \in N(z)$ holds, leading to
$x_W = 0$. The second \ed{inclusion} stems directly from the definition of clique
projection.
\end{proof}

A clique projection of an edge $uv$ is also referred to as {\em edge
projection}. This term is employed in~\cite{Pardalos,Rossi.Smriglio.01}
with a slightly different meaning since, in those papers, the
vertices in $(N(u) \cap N(v)) \cup \{ u, v \}$ are removed when performing the
projection. A fundamental property of edge projection is the following.

\begin{definition}[\cite{Rossi.Smriglio.01}]
An edge $uv \in E$ is {\em projectable} in $G$ if there exists a maximum stable set $S$
in $G$ such that $S \cap \{u, v\} \ne \emptyset$.
\end{definition}

\begin{lemma}[\cite{ManninoSassano96}]
If $uv \in E$ is a projectable edge in $G$, then $\alpha(G) = \alpha(G \mid \{
u, v \})$.
\label{lem:edgeG1}
\end{lemma}

\begin{proof}
Since $uv$ is projectable, we get $\alpha(G) \leq \alpha(G \mid \{ u, v \})$. On
the other hand, $\alpha(G) \geq \alpha(G \mid \{ u, v \})$ due to Lemma~\ref{lem:FSTABG}.
\end{proof}

Results presented in~\cite{ManninoSassano96,Rossi.Smriglio.01} yield that if
$N(u)-\{v\}$ is a clique, then $uv$ is projectable in every induced subgraph of
$G$ containing $u$ and $v$. Indeed, in such a situation the \ed{projection of $uv$}
is equivalent to the projection of the clique $\{ u, v \} \cup (N(u)
\cap N(v))$. Thus, define the subgraph $\tilde G$ as the graph obtained from $G$ by
removing the edges connecting $u$ to \ed{all the vertices} in $N(u) \setminus W$,
where $W \subseteq N(u)-\{v\}$ induces a clique in $G$. Lemma~\ref{lem:edgeG1}
can then be used to write
\[
\alpha(\tilde G \mid \{ u, v \}[H]) = \alpha(\tilde G[H]) \geq \alpha(G[H]),
\]
for every $H \subseteq V$ such that $\{ u, v \} \subseteq H$. A direct
consequence is that the rank inequality $x_H \leq \alpha(\tilde G \mid \{ u, v
\}[H])$, valid for $\tilde G \mid \{ u, v \}$, is also valid for $G$.

\section{The Clique Lifting Operation}
\label{sec:lifting}

In this section we lay a lifting operation that can be applied to valid
inequalities of a projected graph to obtain valid inequalities for $STAB(G)$.
Given an inequality
\begin{equation}
c^\top x = \sum_{v \in H} c_v x_v \leq \beta
\label{eq:validineq}
\end{equation}
with $H \subseteq V$, $\beta \in \mathbb{R}$, and $c \in \mathbb{R}^n$ such
that $c_v \ne 0$ if and only if $v \in H$, we say that $H$ is the {\em support
of \eqref{eq:validineq}}. \ed{We are now in position of stating the lifting lemma
on which our cut-generating procedure is based. For ease of presentation, we will
restrict ourselves to a simplified version of this result in terms of the stable
set polytope, and we refer the reader to \cite{XavierCampelo11} for the general
result. In order to keep this work self-contained, we also provide a proof of this
simplified version.}

\begin{lemma}[Simplified version of the Lifting Lemma~\cite{XavierCampelo11}] \label{lem:xc}
Let $W \subseteq V$ be a clique of $G$. If $c^\top x \leq d$, $c \in
\mathbb{R}^n$ and $d \in \mathbb{R}$, is a valid inequality for $F_W=\{ x \in STAB(G) \mid x_W = 1 \}$
with support $H \subseteq V$, then
\begin{equation}
f(x) = (c^\top x - d) - \lambda\left( x_W - 1 \right) \leq 0,
\label{eq:lifted}
\end{equation}
with \ee{the {\em lifting factor} $\lambda$ being such that}
\begin{equation}
\lambda \leq d - \alpha(G[H \setminus W], c),
\label{eq:lambda}
\end{equation}
is a valid inequality for $STAB(G)$. In addition, if $W$ is a maximal clique,
$c^\top x \leq d$ is facet-defining for $F_W$, and
$\lambda$ satisfies~\eqref{eq:lambda} at equality,
then~\eqref{eq:lifted} is facet-defining for $STAB(G)$.
\label{lem:lifting}
\end{lemma}

\begin{proof}
To \ed{prove} validity, it is sufficient to show that~\eqref{eq:lifted} holds for any
$x \in \mathcal S$. If $x \in F_W$, then $f(x) \leq 0$ holds
because $c^\top x \leq d$ is valid for $F_W$.
Otherwise, $x_W = 0$ and, by definition,
\[
f(x) = (c^\top x - d) + \lambda \leq c^\top x - \alpha(G[H \setminus W], c)
\leq 0.
\]
Now, assume that \ed{$W$ is a maximal clique (so $F_W$ is a facet of $STAB(G)$)
and} $c^\top x \leq d$ is facet-defining for $F_W$, and let $x^1,
\ldots, x^{n-1}$ be $n-1$ affinely independent vectors in $\{ x \in F_W \mid
c^\top x = d \}$. Clearly, $f(x^i) = 0$, for all $i \in \{1,\ldots,n-1\}$.
Additionally, let $x^n$ be the characteristic vector of a maximum weight
independent set of $G[H \setminus W]$ according to $c$. It stems from $\lambda =
d - \alpha(G[H \setminus W], c)$ that
\[
f(x^n) = (c^\top x^n - d) - \lambda\left( x^n_W - 1 \right) = (\alpha(G[H
\setminus W], c) - d) + d - \alpha(G[H \setminus W], c) = 0.
\]
Finally, since $x^n \not\in F_W$, $x^1, \ldots, x^n$ are affinely
independent vectors in $\{ x \in STAB(G) \mid f(x) = 0 \}$.
\end{proof}

\ed{The lifting operation in~\cite{Rossi.Smriglio.01} corresponds to a special
case of Lemma~\ref{lem:xc} in which inequality $c^\top x \leq d$ is a rank inequality of a projected graph's clique with
empty intersection with $W$. More precisely, it includes the projected edge
$uv$, a clique $\tilde W$ in the projected graph $\tilde G$, {\em and} $N(u) \cap N(v)$ such
that $\tilde W \cap N(u) \cap N(v) = \emptyset$, $\tilde W \cap \{ u, v \} =
\emptyset$, and $N(u) \cap N(v)$ is a clique of $\tilde G$ to produce the valid inequality
\begin{equation}
x_{\tilde W} + x_{\{u,v\}} + x_{N(u) \cap N(v)} \leq 2
\label{eq:edgelift}
\end{equation}
for $STAB(G)$.} It is straightforward to check that \ed{Lemma~\ref{lem:xc}}, with
$W = \{ u, v \} \cup (N(u) \cap N(v))$, $c^\top x \leq d$ being the clique inequality of
$\tilde W$, and $\lambda = -1$, establishes that~\eqref{eq:edgelift} is a valid
inequality for $STAB(G)$.

The clique projection operation and the corresponding
clique-lifting operations according to \ed{Lemma~\ref{lem:xc}} lead to
stronger inequalities than those that can be obtained with the edge projection
method proposed in~\cite{Rossi.Smriglio.01}. As an illustration, consider the structure in
Figure~\ref{fig:notrank} and $W = \{ d, e, f \}$. The \ed{projection of $de$} in this
graph \ed{adds the false edge $ac$, and if we then lift the clique $\{ a, b, c \}$ of $G \mid
de$ we get} the rank inequality $x_a + x_b + x_c + x_d + x_e + x_f \leq 2$. The
same inequality is obtained with \ed{Lemma~\ref{lem:xc}} if we take as
$c^\top x \leq d$ the clique inequality of $G + de$ for $\{ a, b, c \}$.
Nevertheless, even in this simple example, there is an inequality that cannot be
derived with the method in~\cite{Rossi.Smriglio.01}. If we take $\{ a, b, c, f
\}$ as the clique inducing set of vertices associated with $c^\top x \leq d$
in \ed{Lemma~\ref{lem:xc}}, then we get $x_a + x_b + x_c + x_d + x_e + 2x_f
\leq 2$ as a valid (indeed, facet-defining~\cite{Campelo.Campos.Correa.08})
inequality for $STAB(G)$.

\begin{figure}[htb]
\centering
        \begin{subfigure}[Not rank inequality.]{
			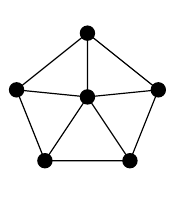
			\label{fig:notrank}}
        \end{subfigure}%
        \qquad
        \begin{subfigure}[Rank inequality.]{
			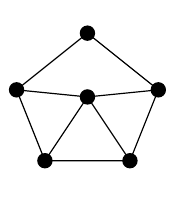
			\label{fig:rank}}
        \end{subfigure}%
\caption{Structures leading to stronger inequalities than edge projection.}
\label{fig:graph}
\end{figure}

The structure in Figure~\ref{fig:rank} (assuming that it induces a rank
inequality of $G$~\cite{Balas.Padberg.76,Campelo.Campos.Correa.08}) also
illustrates the fact that $\sum_{v \in W} x_v \leq 1$ being facet-defining for $STAB(G)$ is not necessary to derive another facet of $STAB(G)$. To show this, we choose $W = \{ d, e \}$ and \ed{again} take the clique inequality of $G + de$ associated with $\{ a, b, c, f \}$. With such a configuration, Lemma~\ref{lem:xc} gives the rank inequality $x_a + x_b + x_c + x_d + x_e + x_f \leq 2$ as well. Observe that this inequality is not derived by the method in~\cite{Rossi.Smriglio.01} if edge $ae$ is deleted before projecting $de$ ($x_b + x_c + x_d + x_e + x_f \leq 2$ would be generated instead).

\section{\ed{Two Procedures for Generating Valid Inequalities}}
\label{sec:procedures}

\ed{We are now in position of introducing the general cut-generating procedures based on the previous definitions and lemmas. We shall introduce two procedures. The first one is a simple procedure directly based on Lemma~\ref{lem:xc}, whereas the second one is a strenghtening based on the results in \cite{XavierCampelo11}. In both procedures, the} generation of a valid inequality consists of the following two steps.

\begin{figure}[htb]
\centering
        \begin{subfigure}[$G_0$ ($= G$ by definition).]{
			\quad\begin{tikzpicture}[thick,fill opacity=0.5,scale=0.45,font=\small]
\foreach \name/\xpos in {2/-6.0,6/-3.0,7/0,8/3.0,5/6.0}
	\node[vertex,ball color=gray!20,opaque] (\name) at (\xpos,3.0) {$\name$};
\foreach \name/\xpos in {3/-1.5,4/1.5}
	\node[vertex,ball color=gray!20,opaque] (\name) at (\xpos,0) {$\name$};
\node[vertex,ball color=gray!20,opaque] (1) at (0,-3.0) {1};

\path[-] 	(8) 	edge (5)
					edge (7)
			(2) 	edge (6)
			(6) 	edge (7)
			(1) 	edge (3)
			(1) 	edge (4)
			(2) 	edge (3)
			(3) 	edge (6)
			(3) 	edge (4)
			(3) 	edge (8)
			(4) 	edge (8)
			(7) 	edge (3)
			(4) 	edge (7);

\draw		(1) 	to[out=150,in=-60] (2);
\draw		(1) 	to[out=30,in=-120] (5);
\draw		(6)		to[out=30,in=150] (8);
\draw		(2) 	to[out=40,in=140] (5);
			\end{tikzpicture}\quad}
        \end{subfigure}\quad\quad
        \begin{subfigure}[$G_3$ ($=G_2$ in this example).]{
        \label{fig:projected}
        \begin{tikzpicture}[thick,fill opacity=0.5,scale=0.45,font=\small]
\shade[rounded corners=1ex,fill=red,above] (1.north west) to[out=150,in=-60]
(2.south east) -- (3.base) -- (3.south) -- cycle;
\node[text=red,text opacity=1] at +(-3,0) {$W_1$};
\shade[rounded corners=1ex,fill=blue,above] (1.north) -- (3.south east) --
(3.base) -- (4.base) -- (4.south west) -- cycle;
\node[text=blue,text opacity=1] at +(0,-1) {$W_2$};
\shade[rounded corners=1ex,fill=green,above] (1.north east) to[out=30,in=-120]
(5.south west) -- (4.base) -- (4.south) -- cycle;
\node[text=olive,text opacity=1] at +(3,0) {$W_3$};

\foreach \name/\xpos in {2/-6.0,6/-3.0,7/0,8/3.0,5/6.0}
	\node[vertex,ball color=gray!20,opaque] (\name) at (\xpos,3.0) {$\name$};
\foreach \name/\xpos in {3/-1.5,4/1.5}
	\node[vertex,ball color=gray!20,opaque] (\name) at (\xpos,0) {$\name$};
\node[vertex,ball color=gray!20,opaque] (1) at (0,-3.0) {1};

\path[-] 	(8) 	edge (5)
					edge (7)
			(2) 	edge (6)
			(6) 	edge (7)
			(1) 	edge (3)
			(1) 	edge (4)
			(2) 	edge (3)
			(3) 	edge (6)
			(3) 	edge (4)
			(3) 	edge (8)
			(4) 	edge (8)
			(7) 	edge (3)
			(4) 	edge (7);

\draw		(1) 	to[out=150,in=-60] (2);
\draw		(1) 	to[out=30,in=-120] (5);
\draw		(6)		to[out=30,in=150] (8);
\draw		(2) 	to[out=40,in=140] (5);

\draw[very thick,color=red]		(5) 	to[out=145,in=35] (6);
\draw[very thick,color=red]		(4) 	-- (5);
\draw[very thick,color=red]		(4) 	-- (6);

\draw[very thick,color=blue]	(2) 	to[out=30,in=150] (7);
\draw[very thick,color=blue]	(5) 	to[out=150,in=30] (7);
\draw[very thick,color=blue]	(2) 	to[out=35,in=145] (8);
\end{tikzpicture}}
        \end{subfigure}
\caption{Projected graph $G_3$ (with false
edges) is obtained after projecting $W_1 = \{ 1, 2, 3 \}$, $W_2 = \{ 1, 3, 4
\}$, and $W_3 = \{ 1, 4, 5 \}$, in this order.}
\label{fig:projection2}
\end{figure}
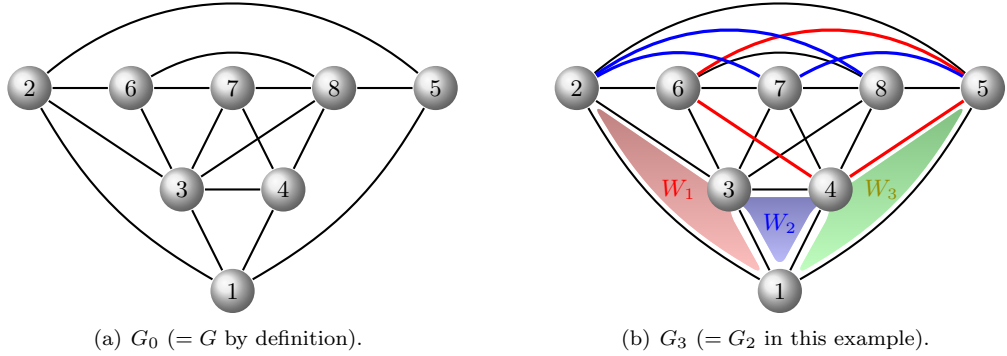

\paragraph{Step 1: Sequence of clique projections} Define $G_0 := G$
and determine distinct subsets $W_1, \ldots, W_r$ of $V$ such that, for
every $t \in \{ 1, \ldots, r \}$, $W_t$ is a clique of $G_{t-1}$ and $G_t$ is
the projected graph $G_{t-1} \mid W_t$. An illustration
of \ed{such a sequence} is depicted in \ed{Figure~\ref{fig:projection2}}.

\paragraph{Step 2: Sequence of clique lifting operations}
\label{sec:clique-lift} Let $W_{r+1} \subseteq V$ be a clique of
$G_r$. The lifting procedure starts with the clique inequality $f_r(x) =
x_{W_{r+1}} \leq 1$, which is valid for $STAB(G_r)$, and iteratively for $t=r-1, \ldots, 0$
\ed{applies \ee{a specific version of} Lemma~\ref{lem:xc} in order to generate}
$f_t(x) = f_{t+1}(x) + \lambda_{t+1} (x_{W_{t+1}} - 1) \leq 1$ in such a way that $f_0(x) \leq 1$ is valid for $STAB(G_0=G)$.

\subsection{Basic Procedure}

\ee{The basic specific version of Lemma~\ref{lem:xc}, stated below, leads the
general method to generate valid inequalities for all projected graphs $G_0,
\ldots, G_r$. For $t \in \{ 1, \ldots, r\}$, let
\[
F_{W_t} = \{ x \in STAB(G_{t-1}) \mid x_{W_t} = 1 \}.
\]}

\begin{lemma}
The inequality
\begin{equation}
x_{W_{r+1}} + \sum_{t=1}^r \lambda^B_t (x_{W_t}
- 1) \leq 1
\label{eq:cliqueproj}
\end{equation}
is valid for $STAB(G)$, where $P_t := \left\{ x \in STAB(G_{t-1}) \mid x_{W_t} =
0 \right\}$ and
\[
\lambda^B_t := \max \left\{ x_{W_{r+1}} + \sum_{i=t+1}^r \lambda^B_i (x_{W_i} -
1) \mid x \in P_t \right\} - 1.
\]
\label{lem:cliqueproj}
\end{lemma}

\begin{proof}
\ed{This result is obtained by iteratively applying Lemma~\ref{lem:xc} on $P_t$,
for $t=r,\dots,1$ (i.e, in reverse order). At step $t$, \ee{the first inclusion
of Lemma~\ref{lem:FSTABG} assures that $x_{W_{t+1}} \leq 1$ is valid for
$F_{W_{t+1}}$. Hence, being the lifting factor $\lambda^B_t$} calculated
according to the definition in \eqref{eq:lambda}, the obtained inequality is valid.}
\end{proof}

\ed{Consider the projected graph $G_3$ in Figure~\ref{fig:projected}, and let
$W_4 = \{ 2, 5, 6, 7, 8 \}$. The inequality $x_{W_4} \leq 1$ is trivially valid
for $STAB(G_3)$, and is also valid} for $STAB(G_2)$ since \ee{$\lambda^B_3 = 0$}.
This comes from \ee{$\lambda^B_3 = \max \left\{ x_{W_4} \mid x_{W_3} = 0 \right\}
- 1$}, which has $\{ 8 \}$ as an \ed{optimal} solution. The remaining lifting
operations are illustrated in Figure~\ref{fig:basiclifting}, \ed{finally giving rise to the inequality $x_{\{4,5,6,7,8\}} + 2x_{\{1,2,3\}} \leq 3$, which is valid for $STAB(G)$.}

\begin{figure}[htbp]
\centering
        \begin{subfigure}[Lifting
        $x_{W_4} \leq 1$ with
        $\lambda^B_2 = \max \left\{ x_{W_4} \mid x_{W_2} = 0 \right\} -
        1 = 1$ results in $x_{W_4} + x_{W_2} \leq 2$, valid for
        $STAB(G_1)$.]{		\label{sfig:cl2}
        \quad\begin{tikzpicture}[thick,fill
        opacity=0.5,scale=0.45,font=\small]
\shade[rounded corners=1ex,fill=red,above] (1.north west) to[out=150,in=-60]
(2.south east) -- (3.base) -- (3.south) -- cycle;
\node[text=red,text opacity=1] at +(-3,0) {$W_1$};
\shade[rounded corners=1ex,fill=blue,above] (1.north) -- (3.south east) --
(3.base) -- (4.base) -- (4.south west) -- cycle;
\node[text=blue,text opacity=1] at +(0,-1) {$W_2$};
\shade[rounded corners=1ex,fill=orange,above] (2.north west) --
(5.north east) -- (5.south east) -- (2.south west) -- cycle;

\foreach \name/\xpos in {2/-6.0,6/-3.0,7/0,8/3.0,5/6.0}
	\node[vertex,ball color=gray!20,opaque] (\name) at (\xpos,3.0) {$\name$};
\foreach \name/\xpos in {3/-1.5,4/1.5}
	\node[vertex,ball color=gray!20,opaque] (\name) at (\xpos,0) {$\name$};
\node[vertex,ball color=gray!20,opaque] (1) at (0,-3.0) {1};

\path[-] 	(8) 	edge (5)
					edge (7)
			(2) 	edge (6)
			(6) 	edge (7)
			(1) 	edge (3)
			(1) 	edge (4)
			(2) 	edge (3)
			(3) 	edge (6)
			(3) 	edge (4)
			(3) 	edge (8)
			(4) 	edge (8)
			(7) 	edge (3)
			(4) 	edge (7);

\draw		(1) 	to[out=150,in=-60] (2);
\draw		(1) 	to[out=30,in=-120] (5);
\draw		(6)		to[out=30,in=150] (8);
\draw		(2) 	to[out=40,in=140] (5);

\node[vertex,ball color=blue,opaque] (2) at (-6.0,3.0) {2};
\node[vertex,ball color=blue,opaque] (7) at (0,3.0) {7};

\draw[very thick,color=red]		(5) 	to[out=145,in=35] (6);
\draw[very thick,color=red]		(5) 	to[out=150,in=30] (7);
\draw[very thick,color=red]		(4) 	-- (5);
\draw[very thick,color=red]		(4) 	-- (6);
\end{tikzpicture}\quad}
        \end{subfigure}\quad\quad
        \begin{subfigure}[{The \ed{optimal} solution $\{ 4,5,6\}$ of $\max
        \left\{ x_{W_4} + x_{W_2} \mid x_{W_1} = 0 \right\}$ yields
         $\lambda^B_1 = 1$ and $x_{W_4} + x_{W_2} + x_{W_1} \leq 3$
        valid for $STAB(G_0)$}.]{ \label{sfig:cl3}
        \quad\begin{tikzpicture}[thick,fill
        opacity=0.5,scale=0.45,font=\small]
\shade[rounded corners=1ex,fill=orange,above] (2.north west) --
(5.north east) -- (5.south west) to[out=-120,in=30] (1.north east) -- (1.north
west) to[out=150,in=-60] (2.south east) -- cycle;
\shade[rounded corners=1ex,fill=red,above] (1.north west) to[out=150,in=-60]
(2.south east) -- (3.base) -- (3.south) -- cycle;
\node[text=red,text opacity=1] at +(-3,0) {$W_1$};

\foreach \name/\xpos in {2/-6.0,6/-3.0,7/0,8/3.0,5/6.0}
	\node[vertex,ball color=gray!20,opaque] (\name) at (\xpos,3.0) {$\name$};
\foreach \name/\xpos in {3/-1.5,4/1.5}
	\node[vertex,ball color=gray!20,opaque] (\name) at (\xpos,0) {$\name$};
\node[vertex,ball color=gray!20,opaque] (1) at (0,-3.0) {1};

\path[-] 	(8) 	edge (5)
					edge (7)
			(2) 	edge (6)
			(6) 	edge (7)
			(1) 	edge (3)
			(1) 	edge (4)
			(2) 	edge (3)
			(3) 	edge (6)
			(3) 	edge (4)
			(3) 	edge (8)
			(4) 	edge (8)
			(7) 	edge (3)
			(4) 	edge (7);

\draw		(1) 	to[out=150,in=-60] (2);
\draw		(1) 	to[out=30,in=-120] (5);
\draw		(6)		to[out=30,in=150] (8);
\draw		(2) 	to[out=30,in=150] (5);

\node[vertex,ball color=red,opaque] (4) at (1.5,0) {4};
\node[vertex,ball color=red,opaque] (6) at (6.0,3.0) {5};
\node[vertex,ball color=red,opaque] (6) at (-3.0,3.0) {6};
\end{tikzpicture}\quad}
        \end{subfigure}
\caption{Basic lifting procedure starting
with the clique inequality of $W_4 = \{ 2, 5, 6, 7, 8 \}$ to generate
$x_{\{4,5,6,7,8\}} + 2x_{\{1,2,3\}} \leq 3$.}
\label{fig:basiclifting}
\end{figure}
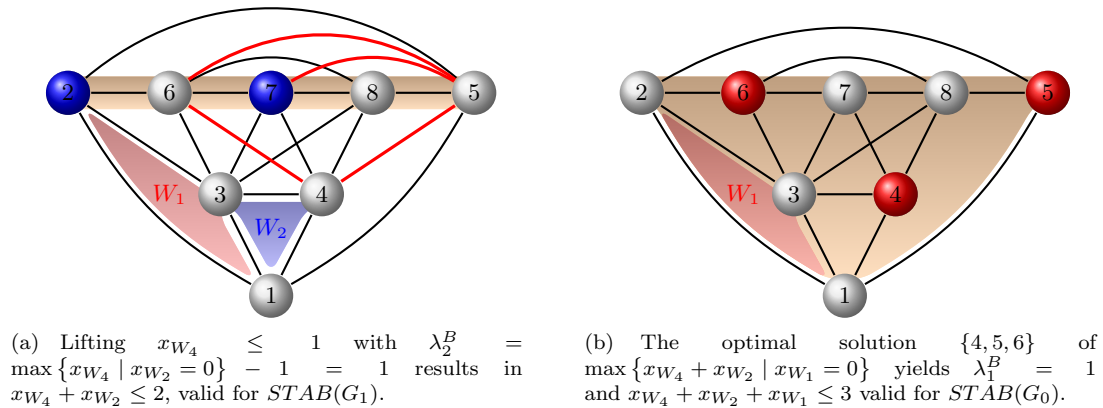

\subsection{Strengthened Procedure}

Let $F_0 := STAB(G)$ and, for $t \in \{ 1, \ldots, r \}$,
\[
F_t = \{ x \in F_{t-1} \mid x_{W_t} = 1 \} = \{ x \in
STAB(G) \mid x_{W_j} = 1, j = 1, \ldots, t \}.
\]
Clearly, the integral elements of $F_t$ are stable sets of $G$. A property
of the clique projection operation is that they are also stable sets of
$G_t$. Define \ed{$E_t$ as the edge set of $G_t$, for $t=1,\dots,r$.}

\begin{lemma}
$F_t \subseteq STAB(G_t)$, for all $t \in \{ 0, \ldots, r \}$.
\label{lem:FtSTABGt}
\end{lemma}

\begin{proof}
We use induction on $t$ to show that $F_t \cap \mathcal S(G) \subseteq STAB(G_t)
\cap \mathcal S(G)$. The \ed{case} $t = 0$ is trivial. For $t \geq 1$,
\ee{$F_{t-1} \subseteq STAB(G_{t-1})$ by the induction hypothesis and,
consequently, $F_t = F_{t-1} \cap F_{W_t}$. The results follow from the first
inclusion of Lemma~\ref{lem:FSTABG}}.
\end{proof}

By definition, $x_{W_t} \leq 1$ is valid for $STAB(G_{t-1})$. The previous lemma
implies that it is also valid for the stable sets of $G$ that \ed{intersect} $W_1, \ldots, W_{t-1}$.

\begin{corollary}
$x_{W_t} \leq 1$ is valid for $F_{t-1}$.
\label{cor:WtvalidFt1}
\end{corollary}


Our strengthened lifting procedure is as
follows. We assume that $\max \emptyset = 0$.

\begin{lemma}
Let $c^\top x \leq d$, $c,x \in \mathbb{R}^n$ and $d \in
\mathbb{R}$, be a valid inequality for $STAB(G_r)$.
Then, $f_t(x) \leq d$ is valid for $F_t$, where, for $t \in \{ 0,
\ldots, r \}$,
\ee{\[
f_t(x) = c^\top x + \sum_{\ell=t+1}^r \lambda^S_\ell (x_{W_\ell} - 1),
\]}
$\mathcal S_t = \mathcal S(G) \cap F_t$, and
\ee{\[
\lambda^S_\ell = \max \Big\{ f_\ell(x) - d \mid
x \in \mathcal S_{\ell-1}, x_{W_\ell} = 0 \Big\}.
\]}
\label{lem:strenglifting}
\end{lemma}

\begin{proof}
We show that $f_t(x) \leq d$ is valid for $F_t$ by induction on $t$. For $t
= r$, the result follows since $f_r(x) = c^\top x \leq d$ is valid for
$STAB(G_r)$ and $F_r \subseteq STAB(G_r)$ by Lemma~\ref{lem:FtSTABGt}.
For $t < r$, by induction hypothesis, $f_{t+1}(x) \leq d$ is valid for
$F_{t+1} = \{ x \in F_t \mid x_{W_{t+1}} = 1 \}$. Applying the inequality
construction, we get \ee{$\lambda^S_{t+1} = 0$} if $\{ x \in \mathcal S_t \mid
x_{W_{t+1}} = 0 \} = \emptyset$, \ed{and}
\ee{\[
\lambda^S_{t+1} = \max \left\{ c^\top x + \sum_{i=t+2}^r \lambda^S_i (x_{W_i} -
1) \mid \begin{array}{l} x \in \mathcal S_t, \\ x_{W_{t+1}} = 0 \end{array} \right\} -
d
\]}
\ed{otherwise.} We now apply Lemma~\ref{lem:lifting}, considering that $x_{W_{t+1}} \leq 1$
is valid for \ee{$F_t$ by Corollary~\ref{cor:WtvalidFt1}} and $f_{t+1}(x) \leq
d$ is valid for $F_{t+1}$, and then obtain that $f_t(x)
\leq d$ is valid for $F_t$.
\end{proof}

Let us take Figure~\ref{fig:stseqlifting} as an example of a sequence of $r = 3$
clique liftings of the projected graph depicted in Figure~\ref{fig:projected}. For
$t=3$, lifting $f_3(x) = x_{W_4} \leq 1$ with $\lambda^S_3 = -1$ generates
$f_2(x) = x_{\{2,6,7,8\}} - x_{\{1,4\}} + 1 \leq 1$, which is valid for $F_2 = \{ x \in
STAB(G) \mid x_{W_1} = x_{W_2} = 1 \}$. \ed{The} iterative procedure of
Lemma~\ref{lem:strenglifting} may generate stronger valid inequalities than the
\ed{basic procedure in Lemma~\ref{lem:cliqueproj}}. For instance, for the example in Figure~\ref{fig:basiclifting}, $x_{\{ 4, 5, 6, 7, 8 \}} +
2x_{\{1, 2, 3 \}} \leq 3$ is a linear combination between $x_{\{1, 2, 4, 6,
7, 8\}} + 2x_{\{3\}} \leq 2$ and the clique inequality $x_{\{1,2,5\}} \leq
1$..

\begin{figure}[htbp]
\centering
        \begin{subfigure}[Lifting
        {$f_2(x) \leq
        1$} with $\lambda^S_2 = 2$ gives {$f_1(x) = x_{\{1,2,4,6,7,8\}} +
        2x_{\{3\}} \leq 2$}, valid for $F_1 = \{ x \in STAB(G) \mid x_{W_1} = 1 \}$.]{ \label{sfig:scl2}
        \quad\begin{tikzpicture}[thick,fill
        opacity=0.5,scale=0.45,font=\small]
\shade[rounded corners=1ex,fill=red,above] (1.north west) to[out=150,in=-60]
(2.south east) -- (3.base) -- (3.south) -- cycle;
\node[text=red,text opacity=1] at +(-3,0) {$W_1$};
\shade[rounded corners=1ex,fill=blue,above] (1.north) -- (3.south east) --
(3.base) -- (4.base) -- (4.south west) -- cycle;
\node[text=blue,text opacity=1] at +(0,-1) {$W_2$};
\shade[rounded corners=1ex,fill=orange,above] (2.north west) --
(5.north east) -- (5.south east) -- (2.south west) -- cycle;

\foreach \name/\xpos in {2/-6.0,6/-3.0,7/0,8/3.0,5/6.0}
	\node[vertex,ball color=gray!20,opaque] (\name) at (\xpos,3.0) {$\name$};
\foreach \name/\xpos in {3/-1.5,4/1.5}
	\node[vertex,ball color=gray!20,opaque] (\name) at (\xpos,0) {$\name$};
\node[vertex,ball color=gray!20,opaque] (1) at (0,-3.0) {1};

\path[-] 	(8) 	edge (5)
					edge (7)
			(2) 	edge (6)
			(6) 	edge (7)
			(1) 	edge (3)
			(1) 	edge (4)
			(2) 	edge (3)
			(3) 	edge (6)
			(3) 	edge (4)
			(3) 	edge (8)
			(4) 	edge (8)
			(7) 	edge (3)
			(4) 	edge (7);

\draw		(1) 	to[out=150,in=-60] (2);
\draw		(1) 	to[out=30,in=-120] (5);
\draw		(6)		to[out=30,in=150] (8);
\draw		(2) 	to[out=40,in=140] (5);

\node[vertex,ball color=blue,opaque] (2) at (-6.0,3.0) {2};
\node[vertex,ball color=blue,opaque] (7) at (0,3.0) {7};

\draw[very thick,color=red]		(5) 	to[out=145,in=35] (6);
\draw[very thick,color=red]		(5) 	to[out=150,in=30] (7);
\draw[very thick,color=red]		(4) 	-- (5);
\draw[very thick,color=red]		(4) 	-- (6);
\end{tikzpicture}\quad}
        \end{subfigure}\quad\quad
        \begin{subfigure}[{$f_0(x) = f_1(x) \leq 2$ is also
        valid for $F_0$ since $\lambda^S_1 = 0$}.]{		\label{sfig:scl3}
        \quad\begin{tikzpicture}[thick,fill
        opacity=0.5,scale=0.45,font=\small]
\shade[rounded corners=1ex,fill=orange,above] (2.north west) --
(5.north east) -- (5.south west) to[out=-120,in=30] (1.north east) -- (1.north
west) to[out=150,in=-60] (2.south east) -- cycle;
\shade[rounded corners=1ex,fill=red,above] (1.north west) to[out=150,in=-60]
(2.south east) -- (3.base) -- (3.south) -- cycle;
\node[text=red,text opacity=1] at +(-3,0) {$W_1$};

\foreach \name/\xpos in {2/-6.0,6/-3.0,7/0,8/3.0,5/6.0}
	\node[vertex,ball color=gray!20,opaque] (\name) at (\xpos,3.0) {$\name$};
\foreach \name/\xpos in {3/-1.5,4/1.5}
	\node[vertex,ball color=gray!20,opaque] (\name) at (\xpos,0) {$\name$};
\node[vertex,ball color=gray!20,opaque] (1) at (0,-3.0) {1};

\path[-] 	(8) 	edge (5)
					edge (7)
			(2) 	edge (6)
			(6) 	edge (7)
			(1) 	edge (3)
			(1) 	edge (4)
			(2) 	edge (3)
			(3) 	edge (6)
			(3) 	edge (4)
			(3) 	edge (8)
			(4) 	edge (8)
			(7) 	edge (3)
			(4) 	edge (7);

\draw		(1) 	to[out=150,in=-60] (2);
\draw		(1) 	to[out=30,in=-120] (5);
\draw		(6)		to[out=30,in=150] (8);
\draw		(2) 	to[out=30,in=150] (5);

\node[vertex,ball color=red,opaque] (4) at (1.5,0) {4};
\node[vertex,ball color=red,opaque] (6) at (6.0,3.0) {5};
\node[vertex,ball color=red,opaque] (6) at (-3.0,3.0) {6};
\end{tikzpicture}\quad}
        \end{subfigure}
\caption{Strengthened lifting procedure starting with the clique inequality of
$W_4 = \{ 2, 5, 6, 7, 8 \}$ to generate $x_{\{1,2,4,6,7,8\}} + 2x_{\{3\}} \leq 2$.}
\label{fig:stseqlifting}
\end{figure}

\subsection{On the Strength of Lemma~\ref{lem:strenglifting}}

\ef{The following results state that the strengthened procedure is, in some
sense, related to facet-subgraphs of $G$. The first result indicates
that, in general, the inequality produced by the strengthened procedure is
stronger than the one produced by the basic procedure.}

\begin{lemma}
\ef{
Let $f_0(x) \leq 1$ be the inequality produced by the strengthened procedure.
Then,
\[
1 + \sum_{\ell=1}^r \lambda^S_\ell = \alpha(G[H^S], c^S)
\leq \alpha(G[H^B], c^B) \leq 1 + \sum_{\ell=1}^r
\lambda^B_\ell,
\]
where $H^S \subseteq \bigcup_{t = 1}^{r+1} W_t$ and $H^B \supseteq H^S$ are
the support, and $c^S$ and $c^B$ are the coefficient vectors, of $f_0(x) \leq 1$
and~\eqref{eq:cliqueproj}, respectively.}

\begin{proof}
\ef{
Let us examine the first equality. The inequality $1 + \sum_{\ell=1}^r
\lambda^S_\ell \geq \alpha(G[H^S], c^S)$ holds since $f_0(x) \leq 1$ is valid by
Lemma~\ref{lem:strenglifting}. For the converse inequality, a stable set of
$H^S$ of weight $1 + \sum_{\ell=1}^r \lambda^S_\ell$ can be constructed by
including a subset of $W_{r+1}$ of weight $\lambda^S_r$ and a vertex of each $W_t$, for all $t \in \{ 0, \ldots, r-1 \}$, by the
definition of $F_{r-1}$ and $\lambda^S_r$.}

\ef{The comparison between the basic and strengthened procedures is given by the
next two inequalities. The former holds because $H^S \subseteq H^B$ and
$\lambda^B_t \geq \lambda^S_t$, for all $t \in \{ 1, \ldots, r \}$, whereas the
validity of~\eqref{eq:cliqueproj} (by Lemma~\ref{lem:cliqueproj}) implies the latter.}
\end{proof}
\end{lemma}

\ef{The second result establishes sufficient conditions for the generated
inequalities to be facet defining. These conditions are slightly weaker than
those in~\cite{XavierCampelo11} due to two differences. First, the subsets $W_3,
\ldots, W_r$ are not required to be cliques of $G$. Second, we use clique
projection and we assume condition~\ref{it:nocliquea} to impose appropriate
false edges in $G_{t'}$ instead of the auxiliary contracted graph defined
in~\cite{XavierCampelo11} to determine $W_{r+1}$. Since the proof is
very similar to the one in that paper, it is left to the appendix.}
\begin{theorem}
\ef{If $f_r(x) = x_{W_{r+1}}$ and
\begin{enumerate}
  \item $|W_t| = k$, the subgraph of $G_{t-1}$ induced by $\bigcup_{i=1}^t
  W_i$ is $k$-partite with vertex classes $V_t^1, \ldots, V_t^k$, and $W_{r+1}$ is a maximal clique of $G_r$ such that $W_{r+1} \cap V_r^k =
  \emptyset$,
  \item $T_t := (V_t, \mathcal W_t)$ is a {\em strong hypertree}
  defined by $V_t := \bigcup_{i=1}^k V_t^i$ and $\mathcal W_t := \{ W_1,
  \ldots, W_t \}$. More precisely, either $\mathcal W_t = \{ V_t \}$ or there is
  a $v \in V_t$ incident to a hyperedge $W_i \in \mathcal W_t$ sharing exactly
  $k - 1$ vertices with some other hyperedge of $T_t$ such that $(V_t \setminus \{v\}, \mathcal W_t \setminus \{W_i\})$ is also a strong hypertree,
  \item for every $i \in \{ 1, \ldots, k-1 \}$ and $w \in V_r^0$ such that
  $N_{G_r}(w) \cap V_r^i \ne \emptyset$, one of the following holds: $v \in W_t \cap V_t^i$ is a
  neighbor of $w$ in $G$ or there exists $t' \in \{ 1, \ldots, r \}$ such that $W_t$ is a clique of
  $G_{t'-1}$, $W_t$ and $W_{t'}$ are adjacent in $T_r$, $v \not\in W_{t'}$,
  and $v' \in W_{t'} \cap V_r^i$ is a neighbor of $w$ in $G_{t'-1}$,
  \label{it:nocliquea}
  \item no $v \in V_t^k$ has neighbors in $V_r^0$, {\em i.e.} $N_{G_r}(v)
  \cap V_r^0 = \emptyset$,
\end{enumerate}
hold for some $k > 0$ and for all $t \in \{ 1, \ldots, r \}$, then $f_t(x) \leq
1$ is facet defining for $F_t$, for all $t \in \{ 1, \ldots, r \}$.}
\label{thm:condii}
\end{theorem}

It can be noticed that the graph and the cliques $W_1, W_2, W_3, W_4$ in
Figure~\ref{fig:projection2} satisfy the sufficient conditions with
$r=3$ and $k=3$. The above theorem implies that the inequality generated by the
strengthened procedure (as shown in Figure~\ref{fig:stseqlifting}) is indeed
facet defining for the graph of the example.

%% file: separat.pdf_t
\begin{picture}(0,0)%
\includegraphics{separat.pdf}%
\end{picture}%
\setlength{\unitlength}{4972sp}%
\begingroup\makeatletter\ifx\SetFigFont\undefined%
\gdef\SetFigFont#1#2#3#4#5{%
  \reset@font\fontsize{#1}{#2pt}%
  \fontfamily{#3}\fontseries{#4}\fontshape{#5}%
  \selectfont}%
\fi\endgroup%
\begin{picture}(1110,1287)(1426,-607)
\put(1981,569){\makebox(0,0)[b]{\smash{{\SetFigFont{11}{13.2}{\rmdefault}{\mddefault}{\updefault}{\color[rgb]{0,0,0}$a$}%
}}}}
\put(1711,-556){\makebox(0,0)[b]{\smash{{\SetFigFont{11}{13.2}{\rmdefault}{\mddefault}{\updefault}{\color[rgb]{0,0,0}$c$}%
}}}}
\put(2251,-556){\makebox(0,0)[b]{\smash{{\SetFigFont{11}{13.2}{\rmdefault}{\mddefault}{\updefault}{\color[rgb]{0,0,0}$d$}%
}}}}
\put(1981,-151){\makebox(0,0)[b]{\smash{{\SetFigFont{11}{13.2}{\rmdefault}{\mddefault}{\updefault}{\color[rgb]{0,0,0}$f$}%
}}}}
\put(1441, 74){\makebox(0,0)[rb]{\smash{{\SetFigFont{11}{13.2}{\rmdefault}{\mddefault}{\updefault}{\color[rgb]{0,0,0}$b$}%
}}}}
\put(2521, 74){\makebox(0,0)[lb]{\smash{{\SetFigFont{11}{13.2}{\rmdefault}{\mddefault}{\updefault}{\color[rgb]{0,0,0}$e$}%
}}}}
\end{picture}%

%% file: nomaximal.pdf_t
\begin{picture}(0,0)%
\includegraphics{nomaximal.pdf}%
\end{picture}%
\setlength{\unitlength}{4972sp}%
\begingroup\makeatletter\ifx\SetFigFont\undefined%
\gdef\SetFigFont#1#2#3#4#5{%
  \reset@font\fontsize{#1}{#2pt}%
  \fontfamily{#3}\fontseries{#4}\fontshape{#5}%
  \selectfont}%
\fi\endgroup%
\begin{picture}(1110,1287)(1426,-607)
\put(1981,569){\makebox(0,0)[b]{\smash{{\SetFigFont{11}{13.2}{\rmdefault}{\mddefault}{\updefault}{\color[rgb]{0,0,0}$a$}%
}}}}
\put(1711,-556){\makebox(0,0)[b]{\smash{{\SetFigFont{11}{13.2}{\rmdefault}{\mddefault}{\updefault}{\color[rgb]{0,0,0}$c$}%
}}}}
\put(2251,-556){\makebox(0,0)[b]{\smash{{\SetFigFont{11}{13.2}{\rmdefault}{\mddefault}{\updefault}{\color[rgb]{0,0,0}$d$}%
}}}}
\put(1981,-151){\makebox(0,0)[b]{\smash{{\SetFigFont{11}{13.2}{\rmdefault}{\mddefault}{\updefault}{\color[rgb]{0,0,0}$f$}%
}}}}
\put(1441, 74){\makebox(0,0)[rb]{\smash{{\SetFigFont{11}{13.2}{\rmdefault}{\mddefault}{\updefault}{\color[rgb]{0,0,0}$b$}%
}}}}
\put(2521, 74){\makebox(0,0)[lb]{\smash{{\SetFigFont{11}{13.2}{\rmdefault}{\mddefault}{\updefault}{\color[rgb]{0,0,0}$e$}%
}}}}
\end{picture}%

%% file: impl-exp.tex
\renewcommand{\algorithmicrequire}{\textbf{Input:}}
\renewcommand{\algorithmicensure}{\textbf{Output:}}
\MakeRobust{\Call}

\section{The Separation Procedure}
\label{sec:separation}

We present in this section a separation procedure based on the
valid inequality generation procedures presented in the last section.
Algorithm~\ref{alg:procedure} summarizes the proposed separation procedure.
Besides the graph $G$, the input of this algorithm is a fractional solution $\bar x$ to be
separated and a set $\mathcal W$ of maximal cliques of $G$. The variable $F$, initially empty,
stores the set of violated inequalities that are generated by the separation
procedure and returned at the end of its execution. For each clique $W_1$ in $\mathcal W$, we
proceed by generating a sequence $S = \langle W_1, \ldots, W_{r+1} \rangle$ of distinct maximal cliques, with the corresponding sequence $\langle
G_0, \ldots, G_r \rangle$ of projected graphs, and a set $T$ of indices $t$ such that
the clique inequality associated with $W_t$ is violated for $STAB(G_{t-1})$. The
generation of the sequence $S$ continues until a certain number of projections
is performed and a violated clique inequality of the current projected graph is found. At this point, all
subsequences $\langle W_1, \ldots, W_t \rangle$ of $S$ defined by a violated
clique are lifted in reverse order as follows: for each $t \in T$,
we apply Lemma~\ref{lem:cliqueproj} or Lemma~\ref{lem:strenglifting}
iteratively to $W_t$ in order to generate a valid inequality for the original
graph. The computation of the lifting factors is accomplished with the
algorithm in~\cite{Ostergard.01}. The set of violated valid inequalities so
generated (stored in $F$) is then returned.

\begin{algorithm}
\caption{Separation procedure: \Call{sepForStab}{$G$, $\bar x$, $\mathcal
W$}}
\label{alg:procedure}
{\footnotesize
\begin{algorithmic}[1]
\Require{Graph $G$, fractional solution $\bar x$ and set $\mathcal W$ of
maximal cliques of $G$}
\Ensure{A set of clique, rank, or weighted rank cuts}
\State $k \gets 0$
\State $F \gets \emptyset$
\State $G_0 \gets G$
\Repeat
	\State $k \gets k + 1$
	\State Select, and remove, a starting clique $W_1$ in $\mathcal W$
	\State $t \gets 0$
	\State $T \gets \emptyset$
	\While {($\bar x_{W_{t+1}} \leq 1 + MINVIOLATION$ or $t \leq MINDEPTH$) and
	$t \leq MAXDEPTH$}
    	\State Project the clique $W_{t+1}$, getting the graph $G_{t+1}$
		\While {$|W_{t+1}| > 2$ and $E_{t+1} = E_t$}
			\State Remove a vertex from $W_{t+1}$
	    	\State Project the clique $W_{t+1}$, getting the graph $G_{t+1}$
	    \EndWhile
    	\If {$E_{t+1} \ne E_t$}
			\If {$\bar x_{W_{t+1}} > 1 + MINVIOLATION$}
				\State $T \gets T \cup \{ t \}$
			\EndIf
	    	\State $t \gets t+1$
		\EndIf
    	\State Find a maximal clique $W_{t+1}$ of $G_t$ \label{lin:maxcliquet}
	\EndWhile
	\ForAll{$t \in T$}
		\State $f_t(x) \gets x_{W_{t+1}}$
		\While {$t > 0$}
			\State $t \gets t - 1$
			\State Compute $\lambda_{t+1}$ on $G_t$ \label{lin:lifting}
    		\State $f_t(x) \gets f_{t+1}(x) + \lambda_{t+1} (x_{W_{t+1}} - 1)$
		\EndWhile
		\If {$f_0(\bar x) > 1 + MINVIOLATION$}
			\State $F \gets F \cup \{ f_0 \}$
		\EndIf
	\EndFor
\Until $\mathcal W = \emptyset$ or $|F| \geq MAXNCUTS$ or $k = MAXITER$
\label{lin:until}
\State \Return $F$
\end{algorithmic}}
\end{algorithm}

The execution of the separation procedure is governed by five parameters. Two
parameters control the number of iterations $k$ of the main loop. According to
the condition checked at line~\ref{lin:until}, the loop is iterated at most
$MAXITER$ times, and this as long as there are cliques left in $\mathcal W$ and
the number of violated inequalities encountered is at most $MAXNCUTS$. At each
iteration, the size $r+1$ of the sequence of projections performed is at least
$MINDEPTH$. The greater is the sequence size $r+1$, the larger is either the
number of variables or the coefficients involved in the valid inequality
generated after the lifting process. An inequality is considered violated only
if its violation is greater than the threshold $MINVIOLATION$. An iteration fails
if no violated clique inequality is found after $MINDEPTH$ clique projections.
The number of projections are bounded from above by parameter $MAXDEPTH$ because
of possible failed iterations, which seldom occurs in practice.

The aim of the set $\mathcal W$ of maximal cliques given as input to the
separation procedure is to yield different sequences of clique projections.
The cliques in $\mathcal W$ are generated with two versions of a greedy
algorithm. In both versions, the generation of a clique consists in selecting an
initial vertex $v$ and then determining a maximal clique in the subgraph induced
by $N(v)$. This clique in $N(v)$ is greedily built by considering the vertices
sorted in a certain order. In one version, vertices are sorted in a nonincreasing order of
weight, where the weight of a vertex $v$ is the value $\bar x_{\{v\}}$. A clique
built with this version tends to have a relatively large intersection with
previous cliques. In the other version, vertices not covered by previous cliques
in $\mathcal W$ have priority with the purpose of generating cliques with small
intersections. In order to avoid repetitions of cliques, the initial vertex of
both versions is one not covered by previous cliques. In order to
maintain a good probability of generating valid inequalities violated by
$\bar x$, only cliques $W$ with $\bar x_W \geq 0.65$ are kept in $\mathcal W$.

Some remarks with respect to the generation of maximal cliques at
line~\ref{lin:maxcliquet} are the following. The heuristic used to generate
the maximal clique $W_{t+1}$ is similar to the one used to generate the cliques
in $\mathcal W$. There are two differences, though: $W_{t+1}$ is guaranteed to
include both a vertex that does not appear in $W_0, \ldots, W_t$ and a false
edge in $E_t \setminus E_{t-1}$, when $t > 0$. For every $K$ clique projections,
we employ the algorithm in~\cite{Tomita.Tanaka.Takahashi.06} to search for a
violated clique. We do not generate all cliques, but stop when a prespecified
number of cliques is enumerated instead. It might be the case that $G_{t+1}$
contains no false edges, relative to $G_t$, which means that no false edges are
generated by the clique projection of $W_{t+1}$. In such a situation, vertices
are iteratively removed from $W_{t+1}$ until either a false edge is generated in
the projection of $W_{t+1}$ or $W_{t+1}$ has only one vertex. In the latter
case, $W_{t+1}$ is discarded.

\section{Computational Experiments and Analysis of Results}
\label{sec:experiments}

In this section we provide some results of computational experiments conducted
in order to explore whether the proposed method is useful as a cut-generating tool
for the MSS problem. Our main goal is not to provide a
competitive algorithm for the MSS problem, since combinatorial algorithms are
much more effective than cutting-plane algorithms for this
problem~\cite{CorreaMCMD14,Segu14}. As already pointed out
in~\cite{Rossi.Smriglio.01}, the facts that other combinatorial
problems may be formulated including stable set constraints, either explicitly
or devised to address their vertex packing relaxation, are motivations to the
investigation of efficient polyhedral methods for the stable set problem. In
this context, we intend to assess whether the proposed procedure is effective at
generating rank or weighted rank cuts for the STAB polytope, and the nature of
the obtained cuts. To this end, we performed three cutting-plane method
implementations attached to the \textsc{COIN-Clp} linear programming solver to
compute a strengthened upper bound for the MSS problem~\cite{coinclp}. In these
implementations, a clique cover of all edges in $E$ is first determined
and the corresponding inequalities constitute the initial model. The
method consists in iteratively solving the current model. Whenever a
fractional solution is found, we first select the set $\mathcal W$ of maximal
cliques of $G$. The violated clique inequalities encountered in this process are
added to the model. Then, the separation procedure of Algorithm~\ref{alg:procedure} is
executed. In addition to the separation procedure, we also implemented the
rounding heuristic proposed in~\cite{Pardalos} and employed it to compute lower bounds.


\begin{table}[htbp]
\centering
{\scriptsize
\begin{tabular}{lcrrrrrrrrr} \hline
\multicolumn{3}{l}{Instance} & & \multicolumn{4}{l}{Upper bound} & \quad &
\multicolumn{2}{l}{Time (sec.)} \\
\cline{1-3} \cline{5-8} \cline{10-11}
Graph & $n$/Dens. & $\alpha$ & LB & $Z_0$ & $SFS\_C$ & $SFS\_B$ & $SFS\_S$ &
\quad & $SFS\_B$ & $SFS\_S$ \\
\hline
\hline
\input{General-Strengthened-invaders-03-001-120-dimacs-tiempo.tex}
\hline \hline
\input{General-Strengthened-invaders-03-001-120-rand-tiempo.tex}
 \hline \hline
\end{tabular}}
\medskip
\caption{Comparison of upper bounds and processing times among $SFS\_C$,
$SFS\_B$, and $SFS\_S$.}
\label{tab:resultsGeneralStrengthenedTiempo}
\end{table}


The first implementation, called $SFS\_C$, only employs the clique cuts
found in the generation of $\mathcal W$ and aims to serve as a reference
for evaluation of the two other implementations. These are implementations that
include a call to \Call{sepForStab}{} after the generation of $\mathcal W$.
The difference between these two implementations is restricted to how the
lifting operation at line~\ref{lin:lifting} is performed. Version $SFS\_B$ uses
the procedure established in Lemma~\ref{lem:cliqueproj} while
Lemma~\ref{lem:strenglifting} is employed in version $SFS\_S$. Various
configurations of the parameters of \Call{sepForStab}{} were tested and
we chose to report the results corresponding to the following setup:
$MINVIOLATION=0.03$, $MINDEPTH=10$, $K=10$, $MAXITER=50$, and $MAXNCUTS=20$.
These were the values that produced the best average results. All the
implementations were written in C++, compiled with {\tt g++ -std=c++11 -m64
-O -fPIC}, and run on a Intel(R) Core(TM) i7-4790K CPU clocked at 4.00GHz.

Table~\ref{tab:resultsGeneralStrengthenedTiempo} summarizes the results of experiments
with some instances from the DIMACS benchmark and for random graphs with
100--300 vertices. The notation $G(n,d)$ specifies random graphs with $n$ vertices and a density of
$d\in[0,1]$, and for these instances we report the average results over five
randomly-generated instances. The experiments were performed on a 64-bit
personal computer, with a time limit of two minutes. The first four columns
contain the instance name, the number of vertices, the graph density, and its
stability number. The following columns contain data for the cutting-plane
method: the column ``LB'' contains the lower bound found by the rounding
heuristic, the columns ``Upper bound'' contain the upper bound obtained with
the three implementations (in addition, the upper bound $Z_0$ corresponding to
the linear relaxation of the initial model is also indicated), and the columns
``Time'' report the total time spent, in seconds.

We can observe the following facts from the data in
Table~\ref{tab:resultsGeneralStrengthenedTiempo}. The graphs on which the
reduction in the upper bound obtained with \Call{sepForStab}{} is significant
when compared to $SFS\_C$ are {\tt brock200\_2}, {\tt brock200\_4}, {\tt
C125.9}, {\tt C250.9} and {\tt DSJC500.5}. For graphs {\tt c-fat200-5} and {\tt
MANN\_a45}, clique cuts are not capable of improving the bound
obtained with respect to $Z_0$. However, the rank and weighted rank cuts added
with \Call{sepForStab}{} made versions $SFS\_B$ and $SFS\_S$ capable of
improving the upper bound, attaining the optimal value in the first case. The
only case where the versions $ SFS\_B$ e $SFS\_S$ do not get better upper bounds
than $SFS\_C$ is the graph {\tt san400\_0.5\_1}. The reason for this
phenomenon is that cliques become large in projected graphs at depth 6 and
beyond, making the calculation of lifting factors very time consuming.
Thus, the time limit is reached before violated inequalities are found. In
Table~\ref{tab:resultsStrengthenedsan400.0.5}, the results are presented with
a depth limit of 3 projections, where we can observe the improvement of the
 bounds with respect to $SFS\_C$. In general, there is a tendency of the version
 $SFS\_S$ to produce upper bounds slightly better than the version $SFS\_B$. In
 particular, the graph {\tt DSJC500.5} is the case where the difference is more pronounced. Unlike the graphs with particular structures,
 random graphs present a homogenous behavior, with both versions $SFS\_B$
 and $SFS\_S$ having better performance than $SFS\_C$, to the advantage of
 version $SFS\_S$. With regard to the comparison of the processing time
 between the versions with projection of cliques, we observed a trend to an
 increase in version $SFS\_S$ with respect to version $SFS\_B$. There are,
 however, 3 exceptions: {\tt p\_hat300-2}, {\tt p\_hat300-3} and {\tt
 san200\_0.7\_2}. In such cases, there is a significant reduction in processing
 time, with slight improvement in the upper bound obtained. This confirms the
 special case of {\tt gen400\_p0.9\_55} and {\tt gen400\_p0.9\_55} where the integer programming approach has performance far superior to combinatorial algorithms. The only combinatorial algorithm
  that solves these graphs is in~\cite{CorreaMCMD14}.

\begin{table}[htbp]
\centering
{\scriptsize
\begin{tabular}{cccccrrcrrcrr} \hline
& \multicolumn{3}{l}{Upper bound} &
\quad & \multicolumn{2}{l}{Time} & \quad & \multicolumn{2}{l}{$SFS\_B$ N. of cuts} &
\quad & \multicolumn{2}{l}{$SFS\_S$ N. of cuts} \\
\cline{2-4} \cline{6-7} \cline{9-10} \cline{12-13}
 LB & $SFS\_B$ & $SFS\_S$ & \cite{Pardalos} && $SFS\_B$ & $SFS\_S$ && Rank
& W-Rank && Rank & W-Rank \\ \hline \hline
  9 & 13.61 & 13.59 & {\bf 13.24} && 113.60 & 120.46 &
& 1437 & 1072 && 1331 & 1121 \\ 
 \hline \hline
\end{tabular}}
\medskip
\caption{Behavior of $SFS\_B$ and $SFS\_S$ with graph {\tt san400\_0.5\_1}
when the depth limit is 3.}
\label{tab:resultsStrengthenedsan400.0.5}
\end{table}

As shown in Table~\ref{tab:resultsGeneralStrengthenedCortesDIMACS}, the procedure is
able to generate a large number of cuts, and provides upper bounds that are
competitive with those generated in \cite{Pardalos}, \cite{Rossi.Smriglio.01},
and~\cite{Giandomenico.Rossi.Smriglio.13} for a representative sample of
benchmark graphs. It is worth remarking that the upper bound obtained with our
approaches is tighter, in comparison the ones from the literature, when the
density is between 40\% and 75\%. The columns ``Upper bound'' contain the upper
bound attained in \cite{Rossi.Smriglio.01}, \cite{Pardalos}, and~\cite{Giandomenico.Rossi.Smriglio.13} in the root node of their branch-and-cut algorithms, respectively. Finally, the last three columns contain the number of generated clique cuts, violated rank inequalities, and violated weighted rank inequalities, respectively. Similarly to existing procedures, our
cut-generating algorithm finds a large number of violated clique inequalities,
and is also able to find many violated rank inequalities. The number of weighted
rank inequalities generated by the procedure is smaller, but nevertheless
provides an interesting set of additional and non-trivial valid inequalities. In
Table~\ref{tab:resultsGeneralStrengthenedCortesRand}, similar results can be
observed for random generated graphs.

Some characteristics of the results presented in
Table~\ref{tab:resultsGeneralStrengthenedCortesDIMACS} have been also observed
in~\cite{Pardalos} in the context of a comparison between the approach adopted
in that paper of combinining several separation heuristics and the one of edge
projection of~\cite{Rossi.Smriglio.01}. It was observed that in some peculiar
cases (notably, the sparse graphs {\tt C125.9} and {\tt C250.9}), the edge
projection alone performed better than the combination of cuts with respect to
the upper bound obtained. The intriguing question that deserves further clarification is how a strategy which involves a number of approximations
(for instance, in the removal of edges and in determining the righthand side of
the cut resulting of the lifting operation) results in stronger cuts. An analysis of the description given in~\cite{Rossi.Smriglio.01} indicates an inaccuracy in
the proposed procedure that may generate non-valid inequalities, although the
authors asserted in a personal communication that the cuts generated in the
reported experiments are verified to be valid. This fact leaves the possibility of generating cuts that, though
valid, are not generated by procedures that ensure the viability of {\em all}
generated cuts. However, the cuts used in~\cite{Rossi.Smriglio.01} can be considered as a
reference of strong cuts for some sparse graphs.

\begin{landscape}

\begin{table}[htbp]
\centering
{\scriptsize
\begin{tabular}{lcrrrrrrrcrrrcrrr} \hline
\multicolumn{3}{l}{Instance} & \quad & \multicolumn{5}{l}{Upper bound} & \quad &
\multicolumn{3}{l}{$SFS\_B$: Number of cuts} & \quad &
\multicolumn{3}{l}{$SFS\_S$: Number of cuts} \\
\cline{1-3} \cline{5-9} \cline{11-13} \cline{15-17}
Graph & $n$/Dens. & $\alpha$ && $SFS\_B$ & $SFS\_S$ & \cite{Pardalos} &
\cite{Rossi.Smriglio.01} &
\cite{Giandomenico.Rossi.Smriglio.13} && Clique & Rank & W-Rank && Clique & Rank
& W-Rank \\ \hline \hline
\input{General-Strengthened-invaders-03-001-120-dimacs-cortes.tex}
 \hline \hline
\end{tabular}}
\medskip
\caption{Upper bounds obtained with $SFS\_B$ and $SFS\_S$, and their comparison
with the ones from~\cite{Pardalos},~\cite{Rossi.Smriglio.01},
and~\cite{Giandomenico.Rossi.Smriglio.13}.}
\label{tab:resultsGeneralStrengthenedCortesDIMACS}
\end{table}

\end{landscape}


\begin{table}[htbp]
\centering
{\scriptsize
\begin{tabular}{lrrrrcrrrcrrr} \hline
\multicolumn{2}{l}{Instance} & \quad & \multicolumn{2}{l}{Upper bound} & \quad &
\multicolumn{3}{l}{$SFS\_B$: Number of cuts} & \quad &
\multicolumn{3}{l}{$SFS\_S$: Number of cuts} \\
\cline{1-2} \cline{4-5} \cline{7-9} \cline{11-13}
Graph & $\alpha$ && $SFS\_B$ & $SFS\_S$ && Clique & Rank & W-Rank && Clique & Rank
& W-Rank \\ \hline \hline
\input{General-Strengthened-invaders-03-001-120-rand-cortes.tex}
 \hline \hline
\end{tabular}}
\medskip
\caption{Upper bounds obtained and number of cuts applied with $SFS\_B$ and
$SFS\_S$.}
\label{tab:resultsGeneralStrengthenedCortesRand}
\end{table}


%% file: General-Strengthened-invaders-03-001-120-dimacs-tiempo.tex
{\tt brock200\_1} & 200/25&21 & 20 & 45.48 & 38.57 & 35.18 & {\bf 34.85} &&
41.76 & 44.86 \\
\rowcolor[gray]{.9} {\tt brock200\_2} & 200/50&12 & 12 & 28.69 & 22.05 & 17.24 & {\bf 16.29} && 120.29 & 120.23 \\
{\tt brock200\_3} & 200/40&15 & 14 & 35.66 & 28.21 & 24.22 & {\bf 23.26} &&
120.22 & 120.47 \\
\rowcolor[gray]{.9} {\tt brock200\_4} & 200/34&17 & 16 & 37.82 & 31.17 & 27.77 & {\bf 26.81} && 109.81 & 117.72 \\
{\tt brock400\_2} & 400/25&12 & 24 & 78.95 & 65.68 & 63.85 & {\bf 63.27} && 120.67 & 120.95 \\
\rowcolor[gray]{.9} {\tt brock400\_4} & 400/25&17 & 23 & 79.53 & 65.99 & 63.61 & {\bf 63.17} && 120.78 & 120.80 \\
{\tt c-fat200-1} & 200/92&12 & 12 & 12 & 12 & {\bf 12} & {\bf 12} && 0.04 & 0.04 \\
\rowcolor[gray]{.9} {\tt c-fat200-2} & 200/84&24 & 24 & 24 & 24 & {\bf 24} & {\bf 24} && 0.02 & 0.02 \\
{\tt c-fat200-5} & 200/57&58 & 58 & 66.66 & 66.66 & {\bf 58} & {\bf 58} && 22.67 & 27.10 \\
\rowcolor[gray]{.9} {\tt c-fat500-1} & 500/96&14 & 14 & 14 & 14 & {\bf 14} & {\bf 14} && 1.68 & 1.65 \\
{\tt c-fat500-10} & 500/81&126 & 126 & 126 & 126 & {\bf 126} & {\bf 126} && 0.97 & 0.97 \\
\rowcolor[gray]{.9} {\tt c-fat500-2} & 500/93&26 & 26 & 26 & 26 & {\bf 26} & {\bf 26} && 0.69 & 0.69 \\
{\tt c-fat500-5} & 500/96&64 & 64 & 64 & 64 & {\bf 64} & {\bf 64} && 0.81 & 0.81 \\
\rowcolor[gray]{.9} {\tt C125.9} & 125/10&34 & 34 & 44.37 & 43.21 & 38.79 & 38.84 && 1.05 & 1.20 \\
{\tt C250.9} & 250/10&44 & 43 & 77.25 & 71.78 & 65.71 & 65.35 && 11.88 & 18.04 \\
\rowcolor[gray]{.9} {\tt DSJC125.1} & 125/90&34 & 34 & 45.28 & 43.22 & {\bf
39.23} & 39.41 && 1.21 & 0.94 \\
{\tt DSJC125.5} & 125/50&10 & 10 & 20.80 & 15.98 & 12.34 & {\bf 11.97} && 32.44
& 20.25 \\
\rowcolor[gray]{.9} {\tt DSJC500.5} & 500/50&13 & 13 & 58.94 & 46.16 & 39.44 &
{\bf 35.54} && 123.20 & 123.29 \\
{\tt gen400\_p0.9\_55} & 400/90&55 & 55 & 93.07 & 55 & {\bf 55} & {\bf 55} && 1.09 & 1.22 \\
\rowcolor[gray]{.9} {\tt gen400\_p0.9\_65} & 400/90&65 & 65 & 103.59 & 65 & {\bf 65} & {\bf 65} && 1.32 & 1.36 \\
{\tt gen400\_p0.9\_75} & 400/90&75 & 75 & 106.54 & 75 & {\bf 75} & {\bf 75} && 2.59 & 2.84 \\
\rowcolor[gray]{.9} {\tt hamming6-4} & 64/65&4 & 4 & 7.57 & 5.33 & 4.28 & {\bf
4} && 0.24 & 0.23 \\
{\tt hamming8-4} & 256/36&16 & 16 & 16 & 16 & {\bf 16} & {\bf 16} && 0.18 & 0.16
\\
\rowcolor[gray]{.9} {\tt keller4} & 171/35&11 & 11 & 25.48 & 14.82 & {\bf 14.24} & {\bf 14.20} && 12.71 & 16.78 \\
{\tt MANN\_a27} & 378/1&126 & 125 & 135 & 135 & 135 & 135 && 0.11 & 0.11 \\
\rowcolor[gray]{.9} {\tt MANN\_a45} & 1035/0.4&345 & 341 & 363 & 363 & {\bf 360}
& {\bf 360} && 1.39 & 1.43 \\
{\tt MANN\_a9} & 45/7&16 & 16 & 18 & 18 & {\bf 18} & {\bf 18} && 0.00 & 0.00 \\
\rowcolor[gray]{.9} {\tt p\_hat300-1} & 300/75&8 & 8 & 24.12 & 16.06 & 12.94 &
{\bf 12.43} && 121.91 & 120.62 \\
{\tt p\_hat300-2} & 300/51&25 & 25 & 45.34 & 34.06 & 33.64 & {\bf 33.50} && 37.18 & 30.29 \\
\rowcolor[gray]{.9} {\tt p\_hat300-3} & 300/26&36 & 35 & 66.67 & 55.67 & 51.95 & {\bf 51.82} && 67.61 & 57.77 \\
{\tt san200\_0.7\_2} & 200/30&18 & 18 & 27.24 & 19.07 & 18.73 & {\bf 18.54} && 20.37 & 6.83 \\
\rowcolor[gray]{.9} {\tt san200\_0.9\_1} & 200/10&70 & 70 & 70.29 & 70 & {\bf 70} & {\bf 70} && 0.03 & 0.04 \\
{\tt san200\_0.9\_2} & 200/10&60 & 60 & 65.81 & 60 & {\bf 60} & {\bf 60} && 0.26 & 0.24 \\
\rowcolor[gray]{.9} {\tt san200\_0.9\_3} & 200/10&44 & 44 & 59.75 & 44 & {\bf 44} & {\bf 44} && 0.17 & 0.23 \\
{\tt san400\_0.5\_1} & 400/50&13 & 9 & 16.41 & 13.60 & 15.20 & 15.18 && 142.37 & 124.81 \\
\rowcolor[gray]{.9} {\tt san400\_0.9\_1} & 400/10&100 & 100 & 119.24 & 100.11 & {\bf 100} & {\bf 100} && 2.57 & 2.87 \\
{\tt sanr200\_0.7} & 200/30&18 & 18 & 41.30 & 33.83 & 30.92 & {\bf 30.21} &&
73.48 & 77.88 \\
\rowcolor[gray]{.9} {\tt sanr200\_0.9} & 200/10&42 & 42 & 64.19 & 59.98 & 54.30
& {\bf 54.18} && 8.98 & 10.77 \\

%% file: General-Strengthened-invaders-03-001-120-rand-tiempo.tex
{\tt $G(100,10)$} & 100/10&31.2 & 31 & 37.56 & 36.25 & {\bf 33.07} & {\bf 33.06} && 0.35 & 0.38 \\
\rowcolor[gray]{.9} {\tt $G(100,20)$} & 100/20&20.2 & 20.2 & 29.28 & 26.53 & 23.87 & 23.86 && 1.65 & 1.67 \\
{\tt $G(100,30)$} & 100/30&15 & 15 & 24.91 & 21.04 & 18.45 & {\bf 18.23} && 3.69 & 4.58 \\
\rowcolor[gray]{.9} {\tt $G(100,50)$} & 100/50& 9.2 & 9.2 & 17.19 & 13.45 &
10.83 & {\bf 10.60} && 7.82 & 6.52 \\
{\tt $G(150,10)$} & 150/10&37.4 & 37 & 51.51 & 48.56 & 44.07 & 43.99 && 2.17 & 2.60 \\
\rowcolor[gray]{.9} {\tt $G(150,20)$} & 150/20&22.4 & 22.4 & 39.84 & 35.37 & 32.20 & 32.06 && 8.10 & 10.35 \\
{\tt $G(150,30)$} & 150/30&16.6 & 16.6 & 32.88 & 27.80 & 24.57 & {\bf 24.19} && 23.49 & 25.22 \\
\rowcolor[gray]{.9} {\tt $G(150,50)$} & 150/50& 10.2 & 10.2 & 23.14 & 18.06 &
14.15 & {\bf 13.53} && 59.00 & 50.63 \\
{\tt $G(200,10)$} & 200/10&41.6 & 41 & 63.99 & 60.86 & 55.11 & 55.04 && 7.12 & 8.07 \\
\rowcolor[gray]{.9} {\tt $G(200,20)$} & 200/20&26 & 25.4 & 50.73 & 44.56 & 40.87 & 40.65 && 23.76 & 30.26 \\
{\tt $G(200,30)$} & 200/30&18 & 17.8 & 40.75 & 34.35 & 31.21 & {\bf 30.49} && 60.53 & 84.12 \\
\rowcolor[gray]{.9} {\tt $G(200,50)$} & 200/50& 11 & 11 & 29.13 & 22.32 & 17.62
& {\bf 16.65} && 115.00 & 112.76 \\
{\tt $G(300,10)$} & 300/10&-- & 44.8 & 90.70 & 82.81 & {\bf 75.81} & 75.91 && 33.74 & 35.91 \\
\rowcolor[gray]{.9} {\tt $G(300,20)$} & 300/20&28.4 & 27.2 & 69.84 & 59.97 & 56.55 & {\bf 56.38} && 74.90 & 84.56 \\
{\tt $G(300,30)$} & 300/30&20.2 & 19.6 & 56.79 & 47.20 & 43.93 & 43.31 && 120.60 & 120.66 \\
\rowcolor[gray]{.9} {\tt $G(300,50)$} & 300/50& 12 & 12 & 40.07 & 30.68 & 24.47
& {\bf 22.90} && 121.00 & 121.29 \\

%% file: General-Strengthened-invaders-03-001-120-dimacs-cortes.tex
{\tt brock200\_1} & 200/25&21 && 35.18 & 34.85 & -- & -- & {\bf 33.59} && 849 &
728 & 2116 && 847 & 752 & 2083 \\
\rowcolor[gray]{.9} {\tt brock200\_2} & 200/50&12 && 17.24 & {\bf 16.29} & 20.99 & 22.01 & 18.27 && 2957 & 591 & 6261 && 2893 & 628 & 6512 \\
{\tt brock200\_3} & 200/40&15 && 24.22 & {\bf 23.26} & -- & -- & 23.55 && 1873 &
468 & 6065 && 2086 & 563 & 6115 \\
\rowcolor[gray]{.9} {\tt brock200\_4} & 200/34&17 && 27.77 & {\bf 26.81} & 29.93
& 30.87 & {\bf 26.77} && 1491 & 531 & 5531 && 1568 & 622 & 5354 \\
{\tt brock400\_2} & 400/25&12 && 63.85 & {\bf 63.27} & 63.84 & 67.66 & -- && 2696 & 843 & 1394 && 2657 & 862 & 1686 \\
\rowcolor[gray]{.9} {\tt brock400\_4} & 400/25&17 && 63.61 & {\bf 63.17} & 63.89 & 67.98 & -- && 2685 & 778 & 1634 && 2706 & 937 & 1998 \\
{\tt c-fat200-1} & 200/92&12 && {\bf 12} & {\bf 12} & 12.71 & 12.86 & -- && -- & -- & -- && -- & -- & -- \\
\rowcolor[gray]{.9} {\tt c-fat200-2} & 200/84&24 && {\bf 24} & {\bf 24} & {\bf 24} & {\bf 24} & -- && -- & -- & -- && -- & -- & -- \\
{\tt c-fat200-5} & 200/57&58 && {\bf 58} & {\bf 58} & 58.89 & 65.25 & {\bf 58}
&& 97 & 843 & 84 && 149 & 866 & 79 \\
\rowcolor[gray]{.9} {\tt c-fat500-1} & 500/96&14 && {\bf 14} & {\bf 14} & {\bf 14} & 14.98 & -- && 107 & 13 & 9 && 72 & 15 & 11 \\
{\tt c-fat500-10} & 500/81&126 && {\bf 126} & {\bf 126} & {\bf 126} & 223.29 & -- && -- & -- & -- && -- & -- & -- \\
\rowcolor[gray]{.9} {\tt c-fat500-2} & 500/93&26 && {\bf 26} & {\bf 26} & 26.97 & 57.78 & -- && -- & -- & -- && -- & -- & -- \\
{\tt c-fat500-5} & 500/96&64 && {\bf 64} & {\bf 64} & 64.70 & 67.08 & -- && -- & -- & -- && -- & -- & -- \\
\rowcolor[gray]{.9} {\tt C125.9} & 125/10&34 && 38.79 & 38.84 & 41.26 & {\bf 37.40} & 37.81 && 61 & 397 & 148 && 67 & 430 & 132 \\
{\tt C250.9} & 250/10&44 && 65.71 & 65.35 & 69.76 & {\bf 58.30} & 63.95 && 389 & 1111 & 399 && 380 & 1093 & 402 \\
\rowcolor[gray]{.9} {\tt DSJC125.1} & 125/90&34 && 39.23 & 39.41 & -- & --
& {\bf 38.22} && 55 & 387 & 149 && 51 & 365 & 147 \\
{\tt DSJC125.5} & 125/50&10 && 12.34 & {\bf 11.97} & -- & -- & 13.21 && 918 &
255 & 4182 && 968 & 261 & 3270 \\
\rowcolor[gray]{.9} {\tt DSJC500.5} & 500/50&13 && 39.44 & {\bf 35.54} & -- &
52.95 & -- && 5556 & 257 & 1066 && 5538 & 299 & 1291 \\
{\tt gen400\_p0.9\_55} & 400/90&55 && {\bf 55} & {\bf 55} & -- & 56.20 & -- && 513 & 254 & 33 && 523 & 257 & 32 \\
\rowcolor[gray]{.9} {\tt gen400\_p0.9\_65} & 400/90&65 && {\bf 65} & {\bf 65} & -- & 65.25 & -- && 598 & 246 & 23 && 620 & 269 & 17 \\
{\tt gen400\_p0.9\_75} & 400/90&75 && {\bf 75} & {\bf 75} & -- & {\bf 75} & -- && 866 & 408 & 55 && 901 & 475 & 55 \\
\rowcolor[gray]{.9} {\tt hamming6-4} & 64/65&4 && 4.28 & {\bf 4} & -- & -- & 4.64 && 291 & 120 & 383 && 310 & 164 & 417 \\
{\tt hamming8-4} & 256/36&16 && {\bf 16} & {\bf 16} & {\bf 16} & {\bf 16} & -- && 173 & -- & -- && 173 & -- & -- \\
\rowcolor[gray]{.9} {\tt keller4} & 171/35&11 && {\bf 14.24} & {\bf 14.20} & 14.83 & 14.95 & 14.29 && 560 & 464 & 1076 && 542 & 535 & 1172 \\
{\tt MANN\_a27} & 378/1&126 && 135 & 135 & -- & 134.86 & {\bf 132.44} && -- & -- & -- && -- & -- & -- \\
\rowcolor[gray]{.9} {\tt MANN\_a45} & 1035/0.4&345 && 360 & 360 & --
& 360 & {\bf 355.86} && 17 & 22 & -- && 17 & 22 & -- \\
{\tt MANN\_a9} & 45/7&16 && 18 & 18 & -- & -- & {\bf 17.11} && -- & -- & -- &&
-- & -- & -- \\
\rowcolor[gray]{.9} {\tt p\_hat300-1} & 300/75&8 && 12.94 & {\bf 12.43} & -- &
-- & 13.45 && 2851 & 141 & 4046 && 2703 & 174 & 4342 \\
{\tt p\_hat300-2} & 300/51&25 && 33.64 & 33.50 & 33.81 & 34.19 & {\bf 30.73} &&
923 & 89 & 1197 && 933 & 82 & 1153 \\
\rowcolor[gray]{.9} {\tt p\_hat300-3} & 300/26&36 && 51.95 & 51.82 & 54.12 &
53.19 & {\bf 49.79} && 995 & 764 & 1327 && 1086 & 759 & 1341 \\
{\tt san200\_0.7\_2} & 200/30&18 && 18.73 & 18.54 & 18.50 & 19.18 & {\bf 18} &&
747 & 265 & 815 && 676 & 259 & 531 \\
\rowcolor[gray]{.9} {\tt san200\_0.9\_1} & 200/10&70 && {\bf 70} & {\bf 70} & {\bf 70} & {\bf 70} & -- && 32 & 22 & 8 && 32 & 22 & 8 \\
{\tt san200\_0.9\_2} & 200/10&60 && {\bf 60} & {\bf 60} & {\bf 60} & {\bf 60} & -- && 163 & 113 & 10 && 148 & 101 & 5 \\
\rowcolor[gray]{.9} {\tt san200\_0.9\_3} & 200/10&44 && {\bf 44} & {\bf 44} & {\bf 44} & 44.80 & -- && 159 & 165 & 20 && 169 & 179 & 25 \\
{\tt san400\_0.5\_1} & 400/50&13 && 15.20 & 15.18 & {\bf 13.24} & 17.14 & -- && 217 & 35 & 57 && 214 & 30 & 57 \\
\rowcolor[gray]{.9} {\tt san400\_0.9\_1} & 400/10&100 && {\bf 100} & {\bf 100} & {\bf 100} & 100.40 & -- && 661 & 241 & 42 && 671 & 264 & 50 \\
{\tt sanr200\_0.7} & 200/30&18 && 30.92 & 30.21 & -- & -- & {\bf 29.45} && 1196
& 593 & 3446 && 1272 & 642 & 3830 \\
\rowcolor[gray]{.9} {\tt sanr200\_0.9} & 200/10&42 && 54.30 & {\bf 54.18}
& -- & -- & 54.52 && 265 & 880 & 287 && 250 & 973 & 332 \\

%% file: General-Strengthened-invaders-03-001-120-rand-cortes.tex
{\tt $G(100,10)$} &31.2 && {\bf 33.07} & {\bf 33.06} && 35.2 & 300.4 & 83.4 && 34.2 & 297.4 & 93.2 \\
\rowcolor[gray]{.9} {\tt $G(100,20)$} &20.2 && 23.87 & 23.86 && 158.2 & 393.8 & 406.4 && 158 & 389.2 & 455 \\
{\tt $G(100,30)$} &15 && 18.45 & {\bf 18.23} && 295 & 289.4 & 1083.2 && 297 & 319.6 & 1073.2 \\
\rowcolor[gray]{.9} {\tt $G(100,50)$} & 9.2 && 10.83 & {\bf 10.60} && 571 &
152.4 & 2066 && 556.4 & 156.6 & 1889.6 \\
{\tt $G(150,10)$} &37.4 && 44.07 & 43.99 && 117 & 559.4 & 217.8 && 117.6 & 573.2 & 224.2 \\
\rowcolor[gray]{.9} {\tt $G(150,20)$} &22.4 && 32.20 & 32.06 && 338.6 & 611.6 & 821.6 && 348.4 & 636.2 & 830.2 \\
{\tt $G(150,30)$} &16.6 && 24.57 & {\bf 24.19} && 684.2 & 465.6 & 2268.4 && 706.4 & 489.6 & 2294.2 \\
\rowcolor[gray]{.9} {\tt $G(150,50)$} & 10.2 && 14.15 & {\bf 13.53} && 1322.8 &
348 & 5012 && 1305.8 & 350.8 & 4670.4 \\
{\tt $G(200,10)$} &41.6 && 55.11 & 55.04 && 242.8 & 858.4 & 320 && 235.2 & 878.6 & 329.2 \\
\rowcolor[gray]{.9} {\tt $G(200,20)$} &26 && 40.87 & 40.65 && 583.4 & 847.2 & 1089.2 && 592.2 & 886 & 1181.2 \\
{\tt $G(200,30)$} &18 && 31.21 & {\bf 30.49} && 1159.6 & 575 & 3255.4 && 1224 & 667.4 & 3824.6 \\
\rowcolor[gray]{.9} {\tt $G(200,50)$} & 11 && 17.62 & {\bf 16.65} && 2861.8 &
558 & 6133.2 && 2856.8 & 560.2 & 6219.2 \\
{\tt $G(300,10)$} &-- && {\bf 75.81} & 75.91 && 520.8 & 1429.2 & 490.2 && 523.4 & 1458 & 480.2 \\
\rowcolor[gray]{.9} {\tt $G(300,20)$} &28.4 && 56.55 & {\bf 56.38} && 1228.4 & 1092.4 & 1223.4 && 1269 & 1120 & 1287.4 \\
{\tt $G(300,30)$} &20.2 && 43.93 & 43.31 && 2166.6 & 669 & 3903.2 && 2378.4 & 740.6 & 4041 \\
\rowcolor[gray]{.9} {\tt $G(300,50)$} & 12 && 24.47 & {\bf 22.90} && 5550.8 &
610 & 3550.2 && 5541.4 & 666.2 & 3707.8 \\

%% file: sufficient.tex
\section{Sufficient Conditions for Faceteness}

Consider the subsets $W_1, \ldots, W_{r+1}$. We show
in this section that if there exists $k > 0$ such that the
following conditions hold, for all $t \in \{ 1, \ldots, r \}$:
\begin{enumerate}[label=(\Roman*),ref=\Roman*,series=conds]
  \item $|W_t| = k$ and the subgraph of $G_{t-1}$ induced by $\bigcup_{i=1}^t
  W_i$ is $k$-partite with vertex classes $V_t^1, \ldots, V_t^k$,
  \label{it:Wk}
  \item $T_t := (V_t, \mathcal W_t)$ is a strong hypertree
  defined by $V_t := \bigcup_{i=1}^k V_t^i$ and $\mathcal W_t := \{ W_1,
  \ldots, W_t \}$,
  \label{it:Tt}
  \item for all $w \in V_t^0 := V \setminus V_t$, there exists $i \in \{ 1,
  \ldots, k \}$ such that $N_{G_{t-1}}(w) \cap V_t^i = \emptyset$,
  \label{it:vV0}
\end{enumerate}
then
\begin{enumerate}[label=(\emph{\roman*}),ref=\emph{\roman*},series=claims]
  \item $x_{W_t} \leq 1$ is facet defining for $F_{t-1}$.
  \label{it:Ft1}
\end{enumerate}
If, in addition to conditions~\eqref{it:Wk}--\eqref{it:Tt}, we
assume that
\begin{enumerate}[label=(\Roman*),ref=\Roman*,resume*=conds]
  \item for every $i \in \{ 1, \ldots, k-1 \}$ and $w \in V_r^0$ such that
  $N_{G_r}(w) \cap V_r^i \ne \emptyset$, one of the following holds: $v \in W_t \cap V_t^i$ is a
  neighbor of $w$ in $G$ or there exists $t' \in \{ 1, \ldots, r \}$ such that $W_t$ is a clique of
  $G_{t'-1}$, $W_t$ and $W_{t'}$ are adjacent in $T_r$, $v \not\in W_{t'}$,
  and $v' \in W_{t'} \cap V_r^i$ is a neighbor of $w$ in $G_{t'-1}$,
  \label{it:noclique}
  \item no $v \in V_t^k$ has neighbors in $V_r^0$, {\em i.e.} $N_{G_r}(v)
  \cap V_r^0 = \emptyset$,
  \label{it:VkV0}
\end{enumerate}
then we prove that
\begin{enumerate}[label=(\emph{\roman*}),ref=\emph{\roman*},resume*=claims]
  \item $f_t(x) \leq 1$ is facet defining for $F_t$,
  \label{it:Ft}
\end{enumerate}
considering that $f_r(x) = x_{W_{r+1}}$ and $W_{r+1}$ is a maximal clique of
$G_r$ such that $W_{r+1} \cap V_r^k = \emptyset$.

\subsection{Proof of~\eqref{it:Ft1}}

The proof of~\eqref{it:Ft1} depends on the dimension of $F_t$, established
next, which in turn depends on conditions~\eqref{it:Wk}--\eqref{it:vV0}.

\begin{lemma}
If~\eqref{it:Wk} holds, then $|W_\ell \cap V^i_t| = 1$, for all $\ell
\in \{ 1, \ldots, t \}$ and $i \in \{ 1, \ldots, k \}$.
\label{lem:interWV}
\end{lemma}

\begin{proof}
Since $W_\ell$ is a clique and $V^i_t$ is a stable set, $|W_\ell \cap V^i_t|
\leq 1$. By condition~\eqref{it:Wk}, $|W_\ell|=k$ and there are at least $k$
stable sets intersecting $W_\ell$. Therefore, $|W_\ell \cap V^i_t| \geq 1$.
\end{proof}

\begin{lemma}[Adapted from Lemma 3.1 of~\cite{XavierCampelo11}]
If~\eqref{it:Wk}--\eqref{it:vV0} hold, then $dim(F_t) = n -
t$.
\label{lem:dimFt}
\end{lemma}

\begin{proof}
Because the incidence matrix of $T_t$ has rank $t$ due to
conditions~\eqref{it:Tt}, it follows that $dim(F_t) \leq n
- t$. To prove that $dim(F_t) \geq n - t$, we exhibit $n-t+1$ affinely independent vectors of
$F_t$. For this purpose, define $x^i$ to be
the incidence vector of $V_t^i$, for every $i \in \{1, \ldots, k\}$. Clearly,
$x^i \in STAB(G)$ and, by Lemma~\ref{lem:interWV}, $x^i_{W_\ell} = 1$ for all $\ell
\in \{ 1, \ldots, t \}$, which
means that $x^i \in F_t$. For every $v \in V_t^0$, define $y_v = x^i + e_v$ where $i \in \{1, \ldots, k\}$ is such that $v$ is not adjacent to any vertex in $V_t^i$, by
condition~\eqref{it:vV0}, and $e_v$ is the incidence vector of $\{ v \}$.
Again, it is easy to see that $y_v \in F_t$. The $|V_t^0| + k = n - (k + t - 1)
+ k = n - t + 1$ points $\{x^i\}^k_{i=1} \cup \{y_v\}_{v\in V_t^0}$ are affinely
independent.
\end{proof}

\begin{theorem}
If~\eqref{it:Wk}--\eqref{it:vV0} hold, then~\eqref{it:Ft1}
holds for all $t \in \{ 1, \ldots, r \}$.
\label{thm:condi}
\end{theorem}

\begin{proof}
By Lemma~\ref{lem:dimFt}, $dim(F_t) = dim(F_{t-1})-1$. Thus, by
Corollary~\ref{cor:WtvalidFt1}, $F_t = \{x \in F_{t-1} \mid x_{W_t} = 1\}$ is a
facet of $F_{t-1}$.
\end{proof}

The reader might observe that conditions~\eqref{it:Wk}--\eqref{it:vV0} are
sufficient for~\eqref{it:Ft} when $t < r$.

\begin{theorem}
If~\eqref{it:Wk}--\eqref{it:vV0} hold for all $t \in \{ 1,
\ldots, r \}$, $f_r(x) = x_{W_r}$, and $d = 1$, then~\eqref{it:Ft} holds for all
$t \in \{ 1, \ldots, r-1 \}$.
\label{thm:simplercondi}
\end{theorem}

\begin{proof}
By induction on $t$. For $t = r$, $x_{W_{r+1}} \leq 1$ is facet defining for
$F_r$ by~\eqref{it:Ft1}. For $t < r$, $x_{W_{t+1}} \leq 1$ is facet defining for
$F_t$ by~\eqref{it:Ft1} and $f_{t+1}(x)$ is facet defining for $F_{t+1} = \{ x \in F_t \mid x_{W_{t+1}} = 1 \}$ by induction hypothesis.
Thus, the result follows by Lemma~\ref{lem:lifting}.
\end{proof}

\subsection{Proof of~\eqref{it:Ft}}

Conditions~\eqref{it:Wk}--\eqref{it:Tt} imply the
following property of $G_r$ and the stable sets of $G$ that cover $W_1,
\ldots, W_r$. A {\em strong hyperpath} is a
strong hypertree with exactly two vertices of degree 1.

\begin{lemma}[Adapted from Lemma 3.2 of~\cite{XavierCampelo11}]
If~\eqref{it:Wk}--\eqref{it:Tt} hold and $x \in F_t$, then $x_{\{
u \}} = x_{\{ v \}}$, for all $i \in \{ 1, \ldots, k \}$ and $\{ u, v \} \subseteq V_t^i$.
\label{lem:xuv}
\end{lemma}

\begin{proof}
Considering that conditions~\eqref{it:Wk}--\eqref{it:Tt} hold, let $W_{t_1},
\ldots, W_{t_q}$ be the strong hyperpath in $T_t$ connecting $u, v$. We
prove the result by induction on $q$. If $q = 2$, then $x_{W_{t_1}}
- x_{W_{t_2}} = x_{\{u\}} - x_{\{v\}} = 0$. Otherwise, $q > 2$. Let $w \in
W_{t_2} \setminus W_{t_1}$. Since $u \not\in W_{t_2}$, we
conclude that $w \in V_{t_2}^i$ by Lemma~\ref{lem:interWV}. Hence, $W_{t_1},
W_{t_2}$ is a strong hyperpath with 2 hyperedges connecting $u,w$, which gives $x_{\{ u \}} = x_{\{ w \}}$. Moreover, $W_{t_p}, \ldots,
W_{t_q}$, for $p = \max \{ j \mid w \in W_{t_j} \}$, is a strong hyperpath with less than $q$
hyperedges connecting two vertices of $V_t^i$. By inductive hypothesis,
$x_{\{ w \}} = x_{\{ v \}}$. Therefore, $x_{\{ u \}} = x_{\{ v \}}$.
\end{proof}

Differently from~\cite{XavierCampelo11}, we determine $W_{r+1}$ in $G_r$
(instead of defining an auxiliary graph). For this purpose, we use the
following property of $G_r$ due to
conditions~\eqref{it:Wk}--\eqref{it:Tt} and~\eqref{it:noclique}.

\begin{lemma}
Let $v \in V_t^i \cap W_t$, for some $i \in \{ 1, \ldots, k-1 \}$, and $w \in
V_r^0$ be such that $V_r^i \cap N_{G_r}(w) \ne \emptyset$.
If~\eqref{it:Wk}--\eqref{it:Tt} and~\eqref{it:noclique} hold, then $vw \in E_r$.
\label{lem:equalVi}
\end{lemma}

\begin{proof}
If $vw \in E$, then the lemma is trivially valid. Otherwise, by
condition~\eqref{it:noclique}, let $t' \in \{ 1, \ldots, r \}$ be such that
$W_t$ is a clique of $G_{t'-1}$, $W_t$ and $W_{t'}$ are adjacent in $T_r$
(considering conditions~\eqref{it:Wk}--\eqref{it:Tt}), $v \not\in W_{t'}$ and
$v' \in W_{t'} \cap V_r^i$ is a neighbor of $w$ in $G_{t'-1}$. In this situation, $W_{t'} \subseteq
N_{G_{t'-1}}(v) \cup N_{G_{t'-1}}(w)$ implies $vw \in E_{t'} \subseteq E_r$ by
the clique projection of $W_{t'}$.
\end{proof}

To show~\eqref{it:Ft}, we still need the following property of the subgraph of
$G_r$ induced by $V_r^0$ and a certain subset of vertices. Notation $\cong$ denotes
the affine isomorphism relation~\cite{XavierCampelo11}.

\begin{lemma}[Adapted from Lemma 3.3 of~\cite{XavierCampelo11}]
If \eqref{it:Wk}--\eqref{it:VkV0} hold for $t=r$, then
$F_r \cong STAB(G_r[V_r^0 \cup R])$ where $R \subseteq V$ is such that
$|R \cap V_r^i| = 1$, for all $i \in \{1, \ldots, k-1\}$, and $R \cap V_r^k
= \emptyset$.
\label{lem:cong}
\end{lemma}

\begin{proof}
In what follows, we denote by $v_i$ the unique vertex in $R \cap
V_r^i$, for all $i \in \{1, \ldots, k-1 \}$.
\begin{description}
\item[{$STAB(G_r[V_r^0 \cup R]) \to F_r$:}] Take a point in $y \in
STAB(G_r[V_r^0 \cup R])$. For each $u \in V$, set $x_u = y_u$ if $u \in
V_r^0$; $x_u = y_{v_i}$, if $u \in V_r^i$, $i \in \{1, \ldots, k-1 \}$; and $x_u
= 1 - \sum^{k-1}_{i=1} y_{v_i}$, if $u \in V_r^k$. We prove
that $x \in F_r$. First, to show that $x \in STAB(G)$, take $uw \in E$. If $u,w
\in V_r^0 \cup R$, then $x_u + x_w \leq 1$ trivially holds. If $u,w \in V_r$, then $x_u + x_w \leq 1$ because $\{ u, w \} \setminus V_r^i \ne
\emptyset$, for all $i \in \{1, \ldots, k \}$. Otherwise, assume without loss of
generality that $u \in V_r \setminus R$ and $w \in V_r^0 \setminus R$. It turns out
that $u \not\in V_r^k$ by condition~\eqref{it:VkV0}. Then, use
Lemma~\ref{lem:equalVi} to conclude that $v_iw \in E_r$ and, consequently,
$x_u + x_w \leq 1$. To show that $x_{W_\ell} = 1$, for $\ell \in \{ 1, \ldots, t
\}$, we use Lemma~\ref{lem:interWV} to write $x_{W_\ell} = \sum^{k-1}_{i=1} y_{v_i} + (1 -
\sum^{k-1}_{i=1} y_{v_i}) = 1$.
\item[{$F_r \to STAB(G_r[V_r^0 \cup R])$:}] Take $x \in F_r$. For each $v
\in V_r^0$, set $y_v = x_v$, and for each $i \in \{ 1, \ldots, k -1 \}$, set
$y_{v_i} = x_{v_i}$. This mapping is injective due to Lemma~\ref{lem:xuv}. Take
$vw \in E[V_r^0 \cup R]$.
It is straightforward to check that $y_v + y_w \leq 1$.
\end{description}
\end{proof}

Claim~\eqref{it:Ft} follows directly from Lemma~\ref{lem:cong} combined with the
Lifting Lemma, as follows.

\begin{theorem}
If \eqref{it:Wk}--\eqref{it:Tt} and~\eqref{it:noclique}--\eqref{it:VkV0} hold,
$f_r(x) = x_{W_{r+1}}$, and $W_{r+1}$ is a maximal clique of $G_r$ such that $W_{r+1} \cap V_r^k = \emptyset$,
then~\eqref{it:Ft} holds for all $t \in \{ 1, \ldots, r \}$.
\label{thm:condii}
\end{theorem}

\begin{proof}
By induction on $t$. For $t = r$, $x_{W_{r+1}} \leq 1$ is facet defining for
$STAB(G_r[V_r^0 \cup R])$, for every $R \subseteq V$ such that
$|R \cap V_r^i| = 1$, for all $i \in \{1, \ldots, k-1\}$, $R
\cap V_r^k = \emptyset$, and $W_{r+1} \subseteq R$. Such an
$R$ exists since $W_{r+1} \cap V_r^k = \emptyset$. Hence, by
Lemma~\ref{lem:cong}, $f_r(x) = x_{W_{r+1}} \leq 1$ is facet defining for $F_r$. For $t < r$,
$x_{W_{t+1}} \leq 1$ is facet defining for $F_t$ by~\eqref{it:Ft1} (notice
that~\eqref{it:VkV0} implies~\eqref{it:vV0} and therefore we can use
Theorem~\ref{thm:condi}) and $f_{t+1}(x)$ is facet defining for $F_{t+1} = \{ x \in F_t \mid x_{W_{t+1}} = 1 \}$ by induction hypothesis.
Thus, the result follows by Lemma~\ref{lem:lifting}.
\end{proof}

%% file: stab.bbl
\begin{thebibliography}{10}

\bibitem{Atamturk200040}
A.~Atamt\"urk, G.~L. Nemhauser, and M.~W.~P. Savelsbergh.
\newblock Conflict graphs in solving integer programming problems.
\newblock {\em European Journal of Operational Research}, 121(1):40 -- 55,
  2000.

\bibitem{Balas.Padberg.76}
E.~Balas and M.~Padberg.
\newblock Set partitioning: a survey.
\newblock {\em SIAM Review}, 18:710--760, 1976.

\bibitem{Bomze99themaximum}
I.~M. Bomze, M.~Budinich, P.~M. Pardalos, and M.~Pelillo.
\newblock The maximum clique problem.
\newblock In {\em Handbook of Combinatorial Optimization}, pages 1--74. Kluwer
  Academic Publishers, 1999.

\bibitem{BritoSantosPoggi15}
S.~S. Brito, H.~G. Santos, and M.~Poggi.
\newblock A computational study of conflict graphs and aggressive cut
  separation in integer programming.
\newblock {\em Electronic Notes in Discrete Mathematics}, 50:355--360, 2015.
\newblock Proc. of the VIII Latin-American Algorithms, Graphs and Optimization
  Symposium.

\bibitem{Campelo.Campos.Correa.08}
M.~Camp{\^e}lo, V.~Campos, and R.~Corr{\^e}a.
\newblock On the asymmetric representatives formulation for the vertex coloring
  problem.
\newblock {\em Discrete Applied Mathematics}, 156(7):1097--1111, 2008.

\bibitem{Cornuejols.08}
G.~Cornu\'ejols.
\newblock Valid inequalities for mixed integer linear programs.
\newblock {\em Mathematical Programming Ser. B}, 112:3--44, 2008.

\bibitem{CorreaMCMD14}
R.~C. Corr{\^{e}}a, P.~Michelon, B.~Le Cun, T.~Mautor, and D.~Delle Donne.
\newblock A bit-parallel russian dolls search for a maximum cardinality clique
  in a graph.
\newblock {\em arxiv}, abs/1407.1209, 2014.

\bibitem{Dukanovic.Rendl.07}
I.~Dukanovic and F.~Rendl.
\newblock Semidefinite programming relaxations for graph coloring and maximal
  clique problems.
\newblock {\em Mathematical Programming}, pages 345--365, 2007.

\bibitem{coinclp}
The COmputational~INfrastructure for Operations Research~Initiative.
\newblock {\em Coin-or linear programming}.
\newblock https://projects.coin-or.org/Clp.

\bibitem{Giandomenico.Rossi.Smriglio.13}
M.~Giandomenico, F.~Rossi, and S.~Smriglio.
\newblock Strong lift-and-project cutting planes for the stable set problem.
\newblock {\em Mathematical Programming Ser. A}, 141(1--2):165--192, 2013.

\bibitem{LiptakLovasz01}
L.~Lipták and L.~Lovász.
\newblock Critical facets of the stable set polytope.
\newblock {\em Combinatorica}, pages 61--88, 2001.

\bibitem{LovaszPlummer86}
L.~Lov{\'a}sz and M.D. Plummer.
\newblock Matching theory.
\newblock {\em Annals of Discrete Mathematics}, 33:544, 1986.

\bibitem{Lovasz.Schrijver.91}
L.~Lovász and A.~J. Schrijver.
\newblock Cones of matrices and set-functions and 0–1 optimization.
\newblock {\em SIAM Journal on Optimization}, 1:166--190, 1991.

\bibitem{ManninoSassano96}
C.~Mannino and A.~Sassano.
\newblock Edge projection and the maximum cardinality stable set problem.
\newblock In {\em DIMACS Ser. Discrete Math. Theoret. Comput. Sci.}, volume~26,
  pages 249--261.

\bibitem{Ostergard.01}
P.~\"Osterg\aa rd.
\newblock A new algorithm for the maximum weight clique problem.
\newblock {\em Nordic Journal of Computing}, 8(4):424--436, 2001.

\bibitem{Pardalos}
S.~Rebennack, M.~Oswald, D.~O. Theis, H.~Seitz, G.~Reinelt, and P.~M. Pardalos.
\newblock A branch and cut solver for the maximum stable set problem.
\newblock {\em Journal of Combinatorial Optimization}, 21:434--457, 2011.

\bibitem{Rossi.Smriglio.01}
F.~Rossi and S.~Smriglio.
\newblock A branch-and-cut algorithm for the maximum cardinality stable set
  problem.
\newblock {\em Operations Research Letters}, 28:63--74, 2001.

\bibitem{Segundo.Losada.Jimenez.11}
P.~San Segundo, D.~Rodríguez-Losada, and A.~Jiménez.
\newblock An exact bit-parallel algorithm for the maximum clique problem.
\newblock {\em Computers \& Operations Research}, 38:571--581, 2011.

\bibitem{Segu14}
P.~San Segundo and C.~Tapia.
\newblock Relaxed approximate coloring in exact maximum clique search.
\newblock {\em Computers \& Operations Research}, 44:185--192, 2014.

\bibitem{Tomita.Kameda.07}
E.~Tomita and T.~Kameda.
\newblock An efficient branch-and-bound algorithm for finding a maximum clique
  with computational experiments.
\newblock {\em Journal of Global Optimization}, 37(1):95--111, 2007.

\bibitem{Tomita.Tanaka.Takahashi.06}
E.~Tomita, A.~Tanaka, and H.~Takahashi.
\newblock The worst-case time complexity for generating all maximal cliques and
  computational experiments.
\newblock {\em Theoretical Computer Science}, 363:28 -- 42, 2006.

\bibitem{WuHao.15}
Q.~Wu and J.-K. Hao.
\newblock A review on algorithms for maximum clique problems.
\newblock {\em European Journal of Operational Research}, 242(3):693 -- 709,
  2015.

\bibitem{XavierCampelo11}
A.~S. Xavier and M.~B. Camp{\^e}lo.
\newblock A new facet generating procedure for the stable set polytope.
\newblock {\em Electronic Notes in Discrete Mathematics}, 37:183--188, 2011.

\end{thebibliography}
